\documentclass[format=acmsmall, review=false]{acmart}
\usepackage{acm-ec-25}
\usepackage{booktabs} 
\usepackage[ruled]{algorithm2e} 

\SetAlFnt{\small}
\SetAlCapFnt{\small}
\SetAlCapNameFnt{\small}
\SetAlCapHSkip{0pt}
\IncMargin{-\parindent}

\usepackage{bbm}
\usepackage{cleveref}
\usepackage{enumitem}
\crefname{subsection}{subsection}{subsections}
\usepackage{booktabs}
\usepackage{tabularx}
\usepackage{array}

\newcommand{\argmax}{\mathrm{argmax}}

\newcommand{\OPT}{\mathsf{OPT}}

\newcommand{\bI}{\mathbbm{1}}
\newcommand{\bE}{\mathbb{E}}

\newcommand{\bR}{\mathbb{R}}

\newcommand{\cP}{\mathcal{P}}
\newcommand{\op}{\overline{p}}
\newcommand{\luid}{$\ell$-up-1-down\xspace}
\newcommand{\cV}{\mathcal{V}}
\newcommand{\zCust}{z^\mathsf{Cust}}
\newcommand{\zObs}{z^\mathsf{Obs}}

\usepackage{color}              
\usepackage{color-edits}
\addauthor{Will}{blue}
\addauthor[Jackie]{jb}{magenta}

\usepackage[english]{babel}
\usepackage[autostyle, english = american]{csquotes}
\MakeOuterQuote{"}

\setcitestyle{authoryear}

\title[]{Personalized Promotions in Practice: Dynamic Allocation and Reference Effects
}

\author{Jackie Baek}
\email{baek@stern.nyu.edu}
\author{Will Ma}
\email{wm2428@gsb.columbia.edu}
\author{Dmitry Mitrofanov}

\begin{abstract}
Partnering with a large online retailer, we consider the problem of sending daily personalized promotions to a userbase of over 20 million customers.
We propose an efficient policy for determining, every day, the promotion that each customer should receive (10\%, 12\%, 15\%, 17\%, or 20\% off), while respecting global allocation constraints.
This policy was successfully deployed to see a 4.5\% revenue increase  during an A/B test, by better targeting promotion-sensitive customers and also learning intertemporal patterns across customers.

We also consider theoretically modeling the intertemporal state of the customer.
The data suggests a simple new combinatorial model of pricing with reference effects, where the customer remembers the best promotion they saw over the past $\ell$ days as the "reference value", and is more likely to purchase if this value is poor.
We tightly characterize the structure of optimal policies for maximizing long-run average revenue under this model---they cycle between offering poor promotion values $\ell$ times and offering good values once.
\end{abstract}

\begin{document}

\begin{titlepage}

\maketitle


\end{titlepage}

\section{Introduction}

Personalized promotions have become increasing prevalent, with Lyft offering them in its app as early as 2019 \citep{shmoys2019how}, Alibaba Livestream Shopping offering highly personalized deals in real-time to its users watching influencers shop \citep{liu2023dynamic}, and Meituan explicitly discriminating in the Bike pass discounts offered to different users \citep{dai2024data}.  These practices are illustrated in \Cref{fig:intro_coupon}, and are generally allowed so long as they are not deceptive and do not violate anti-discrimination or competition rules.  Personalized promotions fall under the grander movement toward gamified shopping, e.g.\ on Temu \citep{zhou2023temu_game}, and AI-agent shopping assistants, e.g.\ on Alibaba Taobao and TMall \citep{alizila2024taobao_tmall_ai}, which provide highly personalized journeys for users on shopping platforms.

\begin{figure}
\centering
\includegraphics[width=\textwidth]{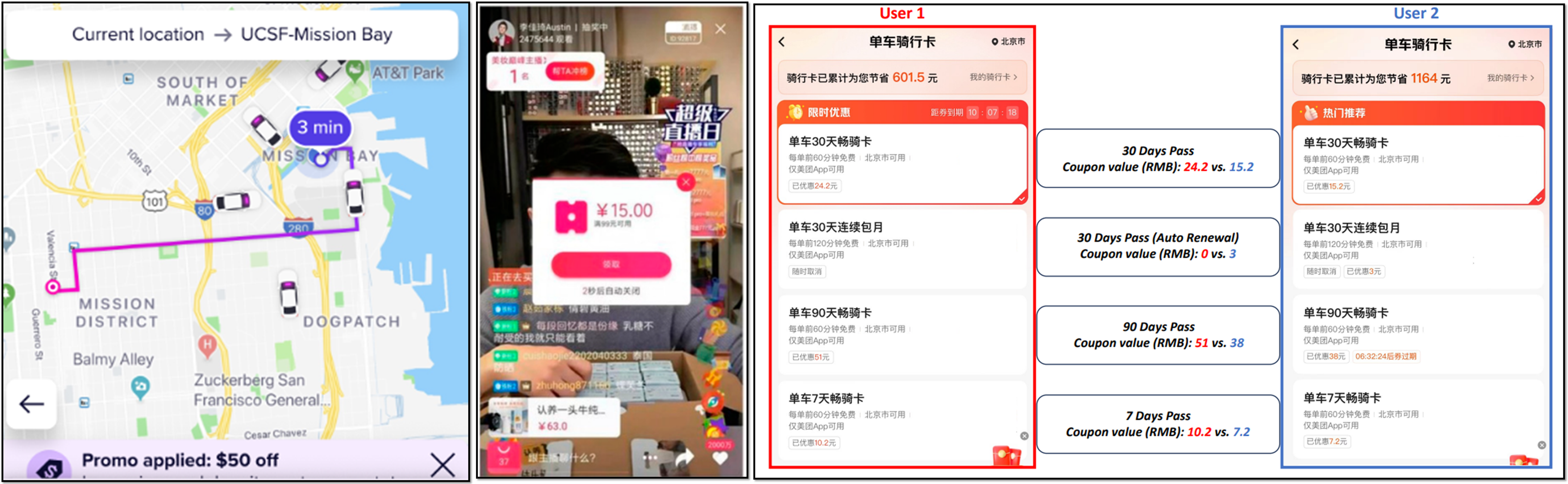}
\caption{Examples of personalized promotions on Lyft, Alibaba Livestream Shopping, and Meituan Bike.}
\label{fig:intro_coupon}
\end{figure}

Optimizing these personalized promotions from the business's end is a challenging problem, due to context-dependent customer behavior, complex interactions with other (non-personalized) promotions running at the same time, and longer-term reference effects where customers may become insensitive to big discounts if they expect even better ones to come.  Moreover, many promotion teams face budget constraints on the total discount redeemed by customers, which creates an allocation problem coupled across customers.  Finally, the algorithm must be fast and scalable, because these promotions are sent at a high-frequency (e.g., daily, or in real-time as the user browses the app), and promotions may need to be computed at the individual user level.

To explore this personalized promotion problem,
we partner with a large U.S.\ online retailer in the home decor sector, whose annual revenue is in the billions.
They send a daily marketing email to their userbase, each containing a personalized coupon offering "X\% off your entire order", where X is either 10\%, 12\%, 15\%, 17\%, or 20\% and targeted toward the specific user.

Our paper has two main parts.
The first, outlined in \Cref{sec:introMethod,sec:introDeploy}, describes our transformation of the coupon algorithm at our partner retailer to use data and optimization, where we develop a fast and scalable algorithm with one "shadow price" parameter that is manually set every day based on the promotional budget that can be allocated.
In an A/B test during May--June 2024 to 20 million customers, we saw a significant 4.5\% uplift in revenue, compared to their incumbent algorithm that sent different discounts based on ad-hoc clustering instead of granular personalization.
From this excellent result,
our algorithm became the default for allocating personalized promotions at our partner retailer in August--September 2024.

The second part of our paper, outlined in \Cref{sec:introTheory}, studies an abstract theoretical model with the goal of improving upon the deployed algorithm that myopically optimizes for next-day revenue.  Based on evidence from the data, we formulate a new model of pricing with reference effects, where the customer remembers the best promotion they saw over the past $\ell$ days as the "reference value".  We combinatorially characterize the structure of optimal policies for this model, which allows us to computationally solve for optimal promotion cycles.

\subsection{Practical Problem and Methodology (details in \Cref{sec:practicalMethodology})} \label{sec:introMethod}

\paragraph{Data and decisions.}
The userbase consists of customers $i$ with static features $\zCust_i$ (e.g., join date, shopping channel).
Every day $t$, the retailer offers each customer $i$ a personalized discount value $v_{it}\in\cV:=\{.10,.12,.15,.17,.20\}$. 
The customer then makes a purchase with pre-discount spend value $w_{it}\ge 0$ ($w_{it} = 0$ represents no purchase),
while also exhibiting auxiliary behavior captured by features $\zObs_{it}$ (e.g., website visits, shopping cart activity, email opens).
The objective is to maximize the long-run revenue
$$
\sum_{t,i} (1-v_{it})w_{it}
$$
while also facing a soft budget constraint on the total discounts redeemed $\sum_{t,i} v_{it}w_{it}$.

\paragraph{Solution approach.}
A customer $i$'s spend $w_{it}$ depends (randomly) on the day $t$, the personalized discount value $v_{it}$, as well as $\zCust_i$ and the entire history of $(v_{it'},w_{it'},\zObs_{it'})_{t'<t}$.
Because the exact spends $w_{it}$ are noisy and difficult to predict, we focus on the binary purchase indicator $y_{it}=\bI(w_{it}>0)$ and normalize spend to $w_{it} = 1$ conditional on purchase.
Maximizing long-run revenue under this assumption is still difficult, and not even formally-defined given the exogenous changes (e.g., the marketing team changing the e-commerce site) that could affect the purchase probabilities on future days.
Therefore, every day $t$ we myopically optimize the daily expected revenue
$$\sum_i(1-v_{it})\cdot\bE[y_{it}|v_{it}].$$ 

A potential concern is that a myopic policy could assign the same discount to a customer repeatedly, a pattern which our partner preferred to avoid.
However, we ensured that our model of $\bE[y_{it}|v_{it}]$ depends sufficiently on past decisions $(v_{it'})_{t'<t}$ to induce natural variation in the offered coupons over time.
Specifically, $\bE[y_{it}|v_{it}]$ becomes less sensitive to $v_{it}$ after a customer has received many strong discounts in recent periods, with customers becoming "complacent" to receiving good discounts (we provide empirical evidence of this in \textbf{\Cref{sec:initialObs}}).
In such cases, the myopic rule naturally shifts the customer toward smaller discounts, generating within-customer variation over time, which our partner viewed as a desirable property for deployment.
This desirable variation, driven by customers’ evolving sensitivity to discounts, inspires the reference effect model studied in the second part of this paper.

\paragraph{Training the model for $\bE[y_{it}|v_{it}]$.}
We extract features to construct a model for $\bE[y_{it}|v_{it}]$, accounting for the expert knowledge of our partner's machine learning team, while also ensuring there are features that depend on past decisions $(v_{it'})_{t'<t}$ to potentially induce the desired variation in discounts received.  To elaborate, we construct a mapping $\phi$ that extracts important context from the feature information and history for a customer $i$ on day $t$, resulting in the (processed) feature
$$
x_{it}:=\phi(t,\zCust_i,(v_{it'},w_{it'},\zObs_{it'})_{t'<t}).
$$
We then solve a supervised learning problem where the goal is to predict $y_{it}$ based on $(x_{it},v_{it})$ over all historical $i,t$ pairs.
We impose structure on the prediction model to isolate how decision variable $v_{it}$ affects the probability that $y_{it}=1$, and find that the direction is intuitively correct (i.e., better discount $v_{it}$ increases purchase probability) in 99.5\% of historical $i,t$ pairs (see \textbf{\Cref{sec:initialObs}}).
This approach lets us pool data across all customers, rather than fitting separate models for each individual, while the processed feature representation enables a flexible mapping from $x$ to purchasing behavior. This approach is inspired by a similar methodology used to personalize healthcare interventions in \citet{baek2025policy}.

\paragraph{Optimization method.}
After training our model, we let $q(x,v)$ denote the probability that it estimates for a customer with context $x$ making a purchase if they are offered discount value $v$.  On day $t$, we then send to each customer $i$
\begin{align} \label{eqn:personalizedObj}
v_{it}\in\argmax_{v\in\cV}(1-\lambda_t v)q(x_{it},v),
\end{align}
where $x_{it}$ is their current context, and $\lambda_t>0$ is a penalty parameter for giving too much discount.
The default value for $\lambda_t$ is 1, in which case~\eqref{eqn:personalizedObj} corresponds to maximizing expected revenue, but $\lambda_t$ can be decreased to give discounts more aggressively, or decreased to conserve the promotional budget.
In practice at our partner retailer, $\lambda_t$ is manually set each day $t$ by an internal employee, who is in touch with upper management on the budget to be allocated across customers.

We highlight that our method is fast and scalable: the prediction function $q(x,v)$ once trained is fast to call, and the optimization problem~\eqref{eqn:personalizedObj} is separately solved for each customer.  It is also easily tunable: $\lambda_t$ can be changed each day to adapt to changing business conditions.

\subsection{Deployment and Impact (details in \Cref{sec:deployResults})} \label{sec:introDeploy}

\paragraph{Collaboration details.}
We were given older data to use academically to develop the algorithm.
After testing on this data, we helped our partner write production-level code both for training a model from their most-recent data and for optimizing daily based on the trained model.
The model is re-trained periodically, although the most-recent data is not shared with us.
This code was deployed in an A/B test during May--June 2024, and aggregate, relative results were shared with us to report.
In addition, a small group of customers were reserved to receive independently random discount values every day, and this data was shared for the purpose of further academic testing.

\paragraph{Results from A/B test.}
In an A/B test on over 20 million customers that were randomly split 50/50 into treatment/control, our algorithm saw a massive 4.5\% increase in average revenue per user.  The $p$-value of such an observation if our algorithm was not better is $<0.01$.
We delve into where the improvement is coming from, and find that our algorithm is able to correctly predict sensitivity to discounts based on a customer's intertemporal state, and thus target the right customers---offer big discounts to customers who need incentive to make a purchase, and send stingy discounts to customers who were going to purchase anyway.
Our algorithm yields a more polarized coupon distribution, where it sends either small or large discounts to most customers, compared to the incumbent algorithm that offered more medium discounts.

In terms of who gets the best discounts, aside from the reference effect considerations, our algorithm tends to offer them to the engaged customers who open the most emails, as it believes that these customers are looking for bargains.  Thus, our algorithm drives the right incentives for customers to engage with the company's emails.  Our algorithm was rolled out to all customers and became the default algorithm for sending personalized promotions during August--September 2024.

\subsection{Theoretical Model and Results (details in \Cref{sec:theory})} \label{sec:introTheory}

Our deployed algorithm optimizes myopically each day, for simplicity, scalability, and to cope with uncontrollable exogenous factors such as arbitrarily changing budgets from day to day.
We now formally study the long-run optimization problem in a theoretical model with a single customer whose intertemporal context is one-dimensional, and no exogenous changes or budgets.
Our model is justified by our data, and 
its optimal solution provides an explanation for coupon cycling in practice as well as provides structure for how it should be optimized.

\paragraph{Reference value model.}
We study the problem of maximizing long-run average expected revenue from a single customer $i$:
\begin{align*}
\sup_{(v_{it})_{t=1}^\infty}\lim_{T\to\infty}\frac1T\sum_{t=1}^T (1-v_{it})q(x_{it},v_{it})    
\end{align*}
with a one-dimensional context $x_{it}:=\max\{v_{i,t-\ell},\ldots,v_{i,t-1}\}$.
This should be viewed as a "reference value", where the customer remembers the best promotion they received over the past $\ell$ days, for some fixed positive integer $\ell$.  It is assumed that this is the sole intertemporal state affecting the customer's purchase probabilities $q(\cdot,v)$ under different discounts $v$, allowing us to construct a deterministic Markov Decision Process (MDP) to model the customer.

We impose only one mild assumption on the function $q(\cdot,v)$, that it is \textit{reference-monotone}: fixing any offered discount $v\in\cV$, the purchase probability $q(x,v)$ is decreasing in $x$.  That is, the bigger the customer's reference value $x$, the more complacent they are to discounts and less likely they are to purchase under any offered $v$.

\paragraph{Model comparison and justification.}
In related literature on dynamic pricing with reference effects (reviewed in \Cref{sec:relatedWork}), the reference price is often defined as the (exponentially-weighted) average of past prices.  Our definition based on the extremum over a fixed memory length $\ell$ yields a more combinatorial model with only $|\cV|$ possible reference values, which will allow us to derive cleaner cycle structure results than most of this literature.

We also provide some empirical evidence for our definition in \textbf{\Cref{sec:empEvid}}, using the data in which customers were sent random discounts.  We find that the reference value $\max\{v_{i,t-\ell},\ldots,v_{i,t-1}\}$ is negatively correlated with purchases ($y_{it}=1$), providing empirical evidence of the reference effect.
We also find that our definition of the reference value is a better univariate predictor for purchases than using the average coupon value over the past $\ell$ days, $\frac{1}{\ell}(v_{i,t-\ell}+\cdots+v_{i,t-1})$.

\paragraph{Characterization of optimal policies.}
For our reference value model, the problem of maximizing long-run average can be formulated as an infinite-horizon undiscounted MDP with $|\cV|^\ell$ states, where the state needs to include the sequence of past $\ell$ values offered.  A priori, this exponential-sized MDP is computationally intractable to solve.
Fortunately, we are able to identify a key "$\ell$-up-1-down" structural result that both provides intuition about optimal cycle structure and allows us to solve this infinite-horizon undiscounted MDP problem in polynomial time.

To elaborate, an "$\ell$-up-1-down" policy is defined by a cycle of \textit{distinct} values in $\cV$, which we call its \textit{generator cycle}.
For each value in the generator cycle, we check if it results in a higher price than the previous value (i.e., the discount is worse)---if so (i.e., price is going "up"), then we repeat the value $\ell$ times; if not (i.e., price is going "down"), then we only offer the value once.
The final policy is to offer this cycle with repeats on the "up" values, ad infinitum.
For example, if $\cV=\{.10,.12,.15,.17,.20\}$ and $\ell=3$, then generator cycle $(.15\ .12\ .20)$ would lead to the policy cycling between values $(.15\ .15\ .15\ .12\ .12\ .12\ .20)$, offering the best discount of $.20$ once a week.  As another example, generator cycle $(.10\ .15\ .12\ .20)$ would lead to the policy cycling between values $(.10\ .10\ .10\ .15\ .12\ .12\ .12\ .20)$.
The intuition behind \luid policies is that if we were going to make the price go up, then the purpose is to "reset" the reference value of the customer, in which case we need to offer the higher price $\ell$ times.

We prove that under the reference-monotonicity assumption, there always exists an \luid policy that is optimal.
This drastically reduces the search space, noting that many cycles are not \luid---e.g.\ if $\ell=3$, then $(.10,.20)$ is not \luid because the .10 is not repeated 3 times;
$(.10,.10,.10,.20,.20,.20)$ is not \luid because the .20 should not be repeated;
and $(.10,.10,.10,.15,.15,.15,.20,.15)$ is not \luid because $(.10,.15,.20,.15)$ is not a valid generator cycle, as it contains the duplicate value $.15$.
To find an optimal policy, we only need to search over generator cycles with distinct values, which we show can be formulated as maximizing infinite-horizon undiscounted reward in a reduced MDP with only $|\cV|$ states, a problem solvable in polynomial time \citep{puterman2014markov}.

\paragraph{Proof technique.}
General theory about finite deterministic MDP's allows for a reduction to stationary deterministic policies defined by cycles, but this is not enough.
Leveraging reference-monotonicity, we further derive a sequence of transformations, none of which worsens the long-run average objective, that allows any policy to be eventually converted into an \luid cycle.
To the best of our understanding, this requires a non-trivial argument that also uses the specific combinatorial structure of our MDP with the "max over past $\ell$" reference value.  Along with the proof, we provide an example illustrating the variety of transformations that may be needed.

\paragraph{Tightness.}
We also prove that \luid is a tight characterization of optimal policies, in that for any generator cycle of distinct values in $\cV$, there exists a reference-monotone instance for which the \luid policy implied by that generator cycle is the unique optimal policy.



\subsection{Related Work} \label{sec:relatedWork}

\paragraph{Dynamic pricing with reference effects.}
Early works used numerical methods to optimize price sequences under reference effects, e.g.\ for peanut butter \citep{greenleaf1995impact}.
Since then, structural results have been established showing the non-necessity of dynamic pricing when customers are loss-averse or loss-neutral \citep{popescu2007dynamic,nasiry2011dynamic}.
Outside of these settings, papers that have tried to formally optimize price sequences found the problem to be generally difficult and exhibit little structure \citep{fibich2003explicit,hu2016dynamic,chen2017efficient}.

By contrast, our theoretical result establishes a clean structure for optimal price sequences, that moreover allows for efficient computation.
This stems from our model being combinatorial than most of these works, using a discrete but general demand function (instead of a specific functional form), a "peak-end" reference model based on the maximum/minimum in recent memory (instead of exponential smoothing), and a long-run steady state objective (i.e., infinite-horizon undiscounted instead of finite-horizon or infinite-horizon discounted).
We note that a combinatorial, graph-theoretic approach to dynamic pricing with reference effects was also taken in \citet{cohen2020efficient}, but they do not derive structural results for optimal policies like we do.

The "peak-end" reference model we study has also been considered in \citet{nasiry2011dynamic,cohen2020efficient,cohen2023high}, with \citet{cohen2023high} establishing optimality of a high-low pricing structure under certain conditions.  Nonetheless, their model is about pricing multiple products under business constraints and quite different in nature from our parsimonious theoretical model, and their structure is also different from our characterization of \luid price sequences.

Finally, our paper provides some data-driven evidence of this "peak-end" reference model in the context of personalized promotions (see \Cref{sec:empEvid}), complementing the aforementioned papers, which reference psychology evidence about people remembering peak experiences \citep{kahneman1993more}.
For earlier papers showing empirical evidence of the reference price effect in general, we refer to the survey by \citet{mazumdar2005reference}.
We should note however that our paper is motivated by \textit{personalized} reference effects, which may be starkly different from estimating reference effects for aggregate demand \citep{hu2018markets}.  \citet{jiang2024intertemporal} consider customer heterogeneity when estimating reference effects.

Recently, the literature on dynamic pricing with reference effects has also incorporated having multiple products with logit demand \citep{guo2025multiproduct}, and demand function learning \citep{den2022dynamic,agrawal2024dynamic}.  At an intuitive level, our \luid policies that offer the unattractive value $\ell$ times to "reset" the reference resembles the ResetRef operation in \citet{agrawal2024dynamic}.

\paragraph{Revenue management with repeated engagements.}
Our "$\ell$-up-1-down" characterization is more similar to the theoretical results derived in the literature on "intertemporal" price discrimination with forward-looking \citep{besbes2015intertemporal} or patient \citep{liu2015optimal} customers---see also \citet{wang2016intertemporal,lobel2020dynamic}.
However, our state is determined by \textit{past} instead of future prices, and our model is in some sense simpler because it is personalized for a single customer instead of dependent on population arrival rates.
This allows for a more exact characterization of optimal price cycles compared to most results in this literature (which only provide a bound on cycle length); we are also able to prove our "$\ell$-up-1-down" characterization is tight (see \Cref{thm:tightness}).

More broadly, our work relates to revenue management under models of repeated engagements with the customer base or an individual customer.
To this end,
\citet{calmon2021revenue} study a model where current prices affect customer goodwill, which in turn affects the future budget they spend on the platform.
\citet{freund2021fair} study a model where users can be given rewards at each round to stay on the platform, but these rewards must satisfy fairness constraints.
\citet{jiang2022revenue} study a model where if the current price is too high, customers can register their willingness-to-pay and be alerted if/when the price drops below their registered price.
\citet{chang2024pricing} study a model where users form usage habits, and pricing can be used to both maximize revenue and curb addiction.

\paragraph{Personalized promotions in practice.}
Our work is based on a real-world deployment of personalized promotions, which can be compared to the deployments in the following works.
Personalized promotion allocation under a promotional budget is formulated as a knapsack problem in \citet{shmoys2019how,albert2022commerce}, which describe real-world applications at Lyft and Booking.com respectively.
Our application is similar, but approach is different, as we propose a single "shadow price" parameter to (approximately) solve the allocation problem instead.  A similar method is in fact deployed to allocate personalized coupons at Meituan \citep{dai2024data}.
Finally, due to the challenges in cleanly modeling intertemporal customer state, \citet{liu2023dynamic} proposes using Deep Reinforcement Learning to black-box learn the best personalized promotion policies for the long-term, which was deployed in real-time at Alibaba Livestream Shopping.
However, we were forced to use a more structured and interpretable approach at our partner retailer, which allows management to see which features cause customers to receive better coupons (see \Cref{sec:initialObs}).

\section{Details of Practical Methodology} \label{sec:practicalMethodology}

We provide details of our practical methodology outlined in \Cref{sec:introMethod}.

We use historical data from 360,000 customers over 90 days.
The raw dataset consists of time-stamped customer-level activity logs that track marketing exposure and subsequent engagement and shopping behavior over time. In particular, the data includes whether and when each customer received marketing emails and coupons, and whether those emails were opened or clicked. The dataset also includes transaction history (purchase occurrence, purchase amounts, and other order-level details) and on-site browsing behavior, such as page views, cart views, and checkout-related activity.
We note that the coupon decisions in the historical dataset were made the by the ad-hoc incumbent algorithm, which does not have access to any features unobservable to us.  Therefore, we do not expect any confounding as long as we control for the observable features.

In collaboration with our partner, we processed the raw data to generate a list of features to use for our estimation model, to capture the majority of the relevant information that may influence a customer's purchasing behavior.
We include information about customers’ 
purchasing and browsing activity (recent purchase volume and transaction counts, site and cart views, and cart composition and discounts), their engagement with emails (opens and clicks over multiple recency windows), their coupon exposure and interaction history (coupon values received, and summary statistics of coupons received, opened, and clicked), as well as a day-of-week indicator. 
Combined, these processed variables yield a total of 61 features, listed in \Cref{sec:practical_method}. 

\subsection{Estimation Model and Method}

Let $\mathcal D$ be the dataset of tuples $(x, v, y)$ of processed features, coupon values, and purchase indicator respectively.
The size of dataset $\mathcal D$ is approximately $360,000\times 90$ (one for each customer, day pair).
We use $\mathcal D$ to estimate $q(x,v)$, the probability that a customer with processed features $x$ makes a purchase when offered coupon value $v$.

One approach to learn $q(x,v)$ would be to fit a flexible black-box machine learning model that maps $(x, v)$ to $y$. 
However, since our downstream goal is to optimize the choice of $v$ rather than to maximize predictive accuracy for $y$, 
we need a reliable estimate of how \textit{changes} in $v$ impact $y$.
A black-box approach can yield unstable or even non-monotone responses to changes in $v$.
Therefore, we impose a structure on $q(x,v)$ that makes the dependence of $v$ explicit, where we decompose the purchase probability $q(x,v)$ into a baseline component $\alpha(x)$ and a coupon-sensitivity component $\beta(x)$. 
Specifically, letting $\sigma(y):=1/(1+\exp(-y))$ denote the logistic function, we impose the structure
\begin{equation}\label{eq:structure}
q(x,v)=\sigma\!\big(\alpha(x) + (v-0.15)\beta(x)\big),
\end{equation}
for learned functions $\alpha(\cdot)$ and $\beta(\cdot)$.
Here, $\alpha(x)$ captures baseline purchase propensity at the "nominal" coupon value of $0.15$, and $\beta(x)$ captures coupon sensitivity: $\beta(x)=0$ implies no effect of the coupon value $v$ on purchase probability, while larger $\beta(x)$ implies stronger responsiveness to changes in $v$.

We note that the structure \eqref{eq:structure} imposes that conditional on $x$, the coupon value $v$ affects purchase probability \textit{only} through the linear term $(v-0.15)\beta(x)$, so the marginal effect of increasing $v$ is governed by a single term $\beta(x)$ that does not vary with $v$. While coupon sensitivity could depend on $v$ in practice, this restriction serves as regularization that enforces a monotone response in $v$ and yields a more stable and interpretable estimate of coupon sensitivity.

We estimate $\alpha(\cdot)$ and $\beta(\cdot)$ using $\mathcal D$ via a two-step procedure: we first learn $\alpha(x)$ flexibly, then learn $\beta(x)$ under a linear structure.
That is, using the subset of $\mathcal D$ consisting only of observations with a coupon value of $v=0.15$, we first fit a gradient-boosting classifier to predict $y$ from $x$. 
Denoting this classifier by $\mathrm{GB}(x)$ as an estimate of $q(x, 0.15)$, and using the fact that $q(x, 0.15) = \sigma(\alpha(x))$, we set $\alpha(x) = \sigma^{-1}(\mathrm{GB}(x))$.

Then, we use the learned $\alpha(x)$ to estimate $\beta(x)$. We impose a linear structure $\beta(x)=\beta^\top x$ (where $\beta,x\in\bR^{61}$), which allows us to write   
\begin{align*} 
    q(x, v) =\sigma\big(\alpha(x) + \beta^{\top}  (v - 0.15) x \big) = \sigma\left( 
    \begin{pmatrix} 
    1 \\ \beta
    \end{pmatrix} ^{\top}
    \begin{pmatrix} 
    \alpha(x) \\
    (v - 0.15)x  
    \end{pmatrix}  
    \right).
\end{align*}
That is, the term inside the logistic function has a linear dependence on the features $\alpha(x)$ and $(v-0.15)x$.
For each sample $(x, v, y) \in \mathcal{D}$, we compute the corresponding features $\alpha(x) = \sigma^{-1}(\mathrm{GB}(x))$ and $(v-0.15)x$, and then we estimate $\beta$ using logistic regression on these transformed features.

\subsection{Empirical Observations of Estimated Model} \label{sec:initialObs}

We fit the model~\eqref{eq:structure} to historical customer-day observations and summarize two empirical patterns.

First,
across historical customer-day pairs, we find that the estimated coupon-sensitivity term $\beta(x)$ is positive for 99.5\% of samples. This is consistent with the basic monotonicity expectation that a larger discount should weakly increase purchase probability.

Next, to understand which features are most strongly associated with $\beta(x)$, we standardize the 61 processed features and fit an auxiliary L1-regularized logistic regression, tuning the penalty to obtain a sparse set of predictors. \Cref{tab:beta_coefs} reports the ten most predictive features and the signs of their coefficients, where a positive sign indicates that larger values of the feature are associated with higher coupon sensitivity.

\begin{table}
\begin{center}
\begin{tabular}{@{}lc@{}}
\toprule
\; Feature & Sign of Coefficient  \\ \midrule
\; (a) \# emails clicked in the last 28 days & $+$ \\ 
\; (b) \# shopping cart visits in the last 3 days   & $+$ \\
\; (c) \# shopping cart visits in the last 7 days   & $+$ \\
\; (d) Average site sale discount in the cart & $+$ \\
\; (e) Average percentage of all historical orders where a coupon was used & $+$ \\
\; (f) Average coupon use percentage of all orders in the last 30 days & $+$ \\
\; (g) Average coupon use percentage of all historical orders & $+$ \\
\; (h) Average coupon discount of a coupon clicked in the last 7 days   & $+$ \\
\; (i) Average coupon discount of a coupon clicked in the last 30 days   & $+$ \\
\; (j) Maximum coupon discount received in the last 7 days & $-$ \\
\bottomrule
\end{tabular}
\end{center}
\caption{
The ten most predictive features for $\beta$.
A positive coefficient implies that a high value for the feature is associated with having a higher sensitivity to the coupon value. 
} \label{tab:beta_coefs}
\end{table}

Features (a)–(c) capture recent engagement, through either email interactions or active cart activity.
The positive coefficients imply that customers who were more recently engaged with emails or their shopping cart are more likely to respond to a better discount. 
Features (d)–(i) capture past reliance on discounts, including frequent coupon use and recent interaction with higher-value coupons. Customers with a high value of these features are ones constantly looking for good discounts, and hence they are not likely to make a purchase with a low-value coupon (i.e., they are price-sensitive customers). Our model associates these features with higher coupon sensitivity.

Finally, feature (j) has a negative sign: customers who recently received a high-value coupon tend to be less responsive to additional increases in $v$. This represents the reference price effect, where receiving a large discount in the recent past decreases the customer's reference price that they need to pay, which makes a subsequent discount less effective.
This mechanism also mitigates repetition in the assigned coupons. If a customer recently received a good discount, their coupon sensitivity decreases, and hence the myopic policy naturally shifts the customer toward worse offers, generating within-customer variation over time.

We provide further evidence of this reference effect in \Cref{sec:empEvid}.
The fact that the \emph{maximum-value} recent coupon is the key predictor of diminished responsiveness motivates our theoretical model in  \Cref{sec:theory}, where the reference value is defined as the maximum discount in recent memory.

\subsection{Optimization Details}

Recall from~\eqref{eqn:personalizedObj} that every day $t$, our algorithm sends to each customer $i$ the discount value $v_{it}\in\cV$ maximizing $(1-\lambda_t v_{it})q(x_{it},v_{it})$, where $\lambda_t$ is a parameter.  We now elaborate on the tuning of $\lambda_t$.  Let $W$ denote the average pre-discount spend in a single shopping cart checkout, over all historical orders.
For any value of $\lambda_t$, we can compute the implied decisions $(v_{it})_i$, under which the expected discount redeemed is
\begin{align} \label{eqn:discRedeemed}
\sum_i v_{it} W q(x_{it},v_{it}).
\end{align}
The value of~\eqref{eqn:discRedeemed} is compared to a promotional budget $B_t$.  If it does not exceed $B_t$, then the discount values $(v_{it})_i$ are sent to the customers.  Otherwise, $\lambda_t$ is increased until~\eqref{eqn:discRedeemed} does not exceed $B_t$.

\begin{proposition}[proven in \Cref{pf:simple}] \label{prop:simple}
Suppose $q(x_{it},v)$ is weakly increasing in $v$ for all $i$.  Then~\eqref{eqn:discRedeemed} is weakly decreasing in parameter $\lambda_t$.
\end{proposition}

Recall from \Cref{sec:initialObs} that the 
predicted purchase probability $q(x_{it},v)$ is increasing in $v$ for 99.5\% of contexts $x_{it}$, due to the non-negative sign of $\beta(x_{it})$.
Therefore, we can essentially think of~\eqref{eqn:discRedeemed} as being decreasing in $\lambda_t$ in practice.

At our partner retailer, an internal employee first evaluates~\eqref{eqn:discRedeemed} under $\lambda_t=1$ (which would maximize revenue), comparing it to the promotional budget $B_t$ for the day $t$.  If~\eqref{eqn:discRedeemed} exceeds $B_t$, then they use bisection search to increase $\lambda_t$ to the smallest value for which~\eqref{eqn:discRedeemed} is no greater than $B_t$.

We note that we tried to get a better prediction of the total discount redeemed by having $W$ depend on the customer $i$ or the discount value $v_{it}$.  
However, this failed because exact spend amounts are highly idiosyncratic, which is why we focus on predicting purchase probabilities instead.

\section{Details of Deployment and Impact} \label{sec:deployResults}

An A/B test was ran during May--June 2024, in which 20 million customers were randomly split 50/50 into treatment/control.
Customers in the treatment group received personalized discount values determined by our algorithm, while
customers in the control group received discount values determined by the incumbent algorithm that had been in production for a couple of years.
Once a customer was assigned to either the treatment or control groups, their assignment remained unchanged throughout the experiment.

\paragraph{Overview of incumbent algorithm.}
The incumbent algorithm focused on what is the correct \textit{distribution} of coupons, i.e.\ what fraction of customers should receive each of the discounts 10\%, 12\%, 15\%, 17\%, and 20\%, which could depend on the promotional budget.
Meanwhile, the customers were clustered in an ad-hoc fashion, based mostly on tenure and average spend.
Customers in the same cluster would generally be sent the same discount value, with arbitrary adjustments as needed to fit the desired coupon distribution.
Additional heuristics were inserted to ensure that each customer saw variation in the discounts they received over time.

Overall, the incumbent algorithm processed data manually, to create coarse clusters of customers, based on static features.  By contrast, our algorithm used machine learning on the data, to create personalized prediction models for each customer, that also accounted for their intertemporal state including reference effects.  To fairly compare our algorithm to the incumbent, the parameter $\lambda_t$ in our algorithm was tuned daily to match the promotional budget of the incumbent algorithm.

\paragraph{Results.}
Aggregate, relative results over an 11-day period from the A/B test were shared with us, and displayed in \Cref{Figure:revenue improvement}.  During this 11-day period, our algorithm's prediction model was not re-trained.  Our algorithm optimizes for revenue, and we indeed saw a 4.5\% higher total revenue over the 11-day period in the treatment group compared to the control group, which have the same size.
We also estimate the average treatment effect by regressing each customer's 11-day total revenue on a treatment indicator and conducting a one-sided test using heteroskedasticity-robust standard errors, and observe a $p$-value less than 0.01 for the null that our algorithm does not increase average revenue.

\begin{figure}
\centering
\includegraphics[width=\textwidth]{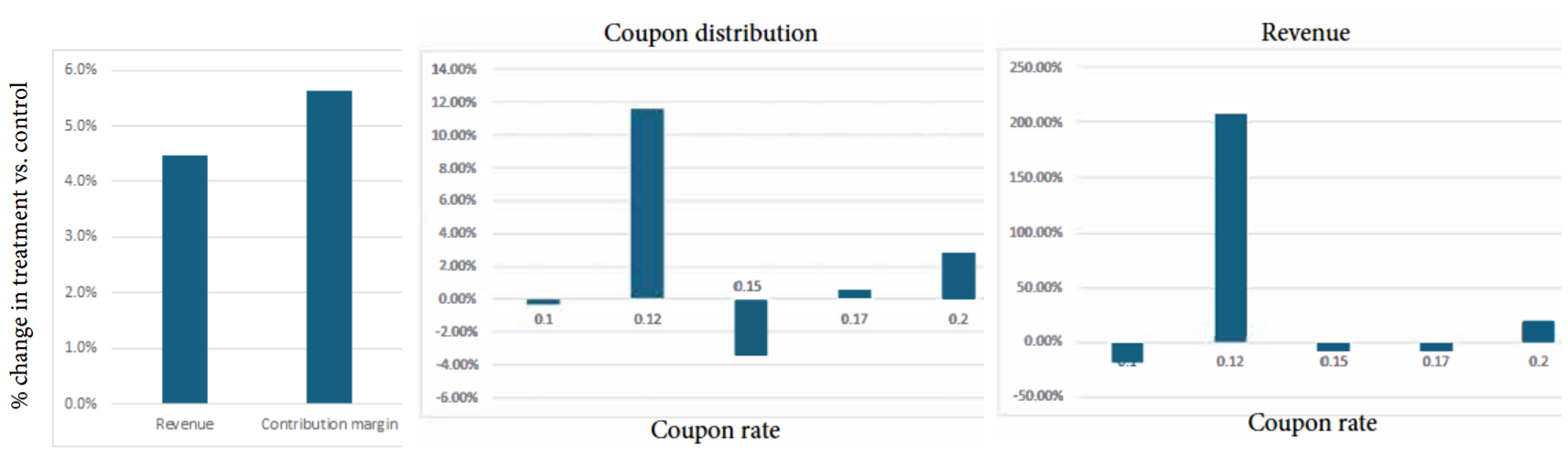}
\caption{Deployment results:
Left chart is overall \% change in revenue and profit between treatment vs.\ control;
Middle chart is \% change in total number of times each coupon rate was sent;
Right chart is \% change in total revenue from coupons of each rate being redeemed.
}
\label{Figure:revenue improvement}
\end{figure}

Our algorithm does not account for cost of goods sold (i.e., the wholesale price our partner retailer paid for its inventory), but we see an even bigger 5.6\% increase in contribution margin (i.e., profit).  This suggests that our coupon targeting could be enticing the customers to substitute to higher-end goods which have higher margins, although we did not formally test this.

Regardless, the stark results on the Left chart of \Cref{Figure:revenue improvement} fully demonstrate the power of machine learning, granular personalization, and intertemporal state---three new elements that our algorithm introduced to our partner's business.
We also emphasize that these numbers pertain to overall revenue and contribution margin, without any subdivision of the goods or customers, which was very attractive to our partner's management.
We completely changed how they think about allocating coupons, where instead of directly optimizing the distribution of coupon values sent, one should learn individual customer models and personalize coupon values, with the distribution of coupons being a by-product of the individual-level optimization.

\paragraph{Where is the money coming from?}
The Middle chart of \Cref{Figure:revenue improvement} shows how our algorithm changed the distribution of coupons.  In particular, there is an 11.56\% increase in the number of 0.12 coupons sent, and a 2.91\% increase in the number of .2 coupons sent, mostly at the expense of a 3.44\% decrease in the number of .15 coupons sent.
Overall, our algorithm seems to think that the 0.12 coupon rate strikes a nice balance, where it is not as unattractive to the customer as the 0.1 rate, but preserves more revenue for the retailer than the 0.15 rate.

The Right chart of \Cref{Figure:revenue improvement} shows massive changes in the total revenue from coupons of each rate being redeemed, suggesting that our algorithm is significantly shifting which customers receive each coupon rate, and overall enticing more purchases.  In particular, there is a 208\% higher (i.e., tripled) revenue from 0.12 coupons being redeemed, with only an 11.56\% increase in the number of 0.12 coupons sent, suggesting that our algorithm is correctly targeting this smaller discount toward customers who were going to make a purchase regardless of the coupon rate.
At the same time, there is a 20.5\% higher revenue from 0.2 coupons being redeemed, with only a 2.91\% increase in the number of 0.2 coupons sent, suggesting that our algorithm is also correctly targeting the best discount at customers who need the bargain to make a purchase.

\paragraph{Which customers receive the best coupons?}
Generally, the structure of our estimation and optimization is such that every day $t$, customers $i$ in an intertemporal state $x_{it}$ with a higher value of $\beta^\top x_{it}$ tend to receive a bigger discount $v_{it}$.  Looking back at \Cref{tab:beta_coefs}, this suggests that engaged customers who have redeemed more coupons in the past would receive better coupons in the future.  However, this may not happen if they have already received the biggest coupon in the past 7 days, inducing some natural cycling in coupons received even for these customers.

\subsection{Empirical Evidence for Theoretical Model} \label{sec:empEvid}

We provide empirical support for the key ingredients of our theoretical model (\Cref{sec:theory}), including reference effects and the reference-monotonicity assumption.
We leverage data from a group of 150,000 customers who, for an 11-day period during the A/B test, were assigned an independent random coupon each day, drawn uniformly from $\{0.12,0.15,0.17,0.20\}$. The randomized assignment allows us to directly measure correlations without controlling for any factors.

\paragraph{Evidence of reference effect.}
For various memory lengths $\ell$, we define the reference value for a customer $i$ on day $t$ as $\max\{v_{i,t-\ell},\ldots,v_{i,t-1}\}$, and compute its empirical correlation with the purchase indicator $y_{it}$ across customer-day observations from the randomized dataset. \Cref{fig:ref_corr} shows that this correlation is generally negative across $\ell$, suggesting that a bigger recent-best discount (a bigger reference value) is associated with a lower propensity to purchase, consistent with a reference effect under our definition.

As a comparison, we also consider an average-based reference metric, $\frac{1}{\ell}(v_{i,t-\ell}+\cdots+v_{i,t-1})$.
The corresponding correlations are not as strong or negative, which suggests that the maximum over recent history is a more predictive univariate proxy for intertemporal effects than the average.

\begin{figure}
\centering
\includegraphics[width=0.5\textwidth]{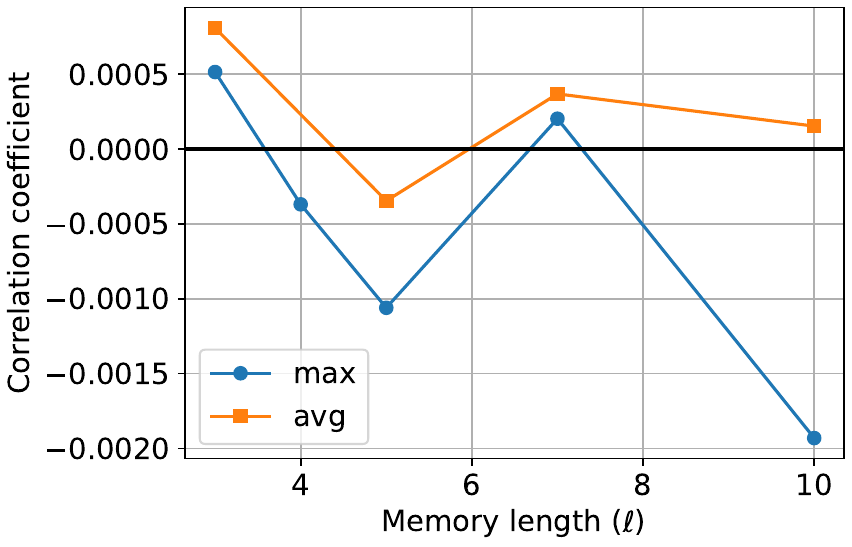}
\caption{Correlation between reference coupon metrics and the purchase indicator $y_{it}$ for various memory lengths $\ell$. The reference metric ``max'' uses $\max\{v_{i,t-\ell},\ldots,v_{i,t-1}\}$, and ``avg'' uses $\frac{1}{\ell}(v_{i,t-\ell}+\cdots+v_{i,t-1})$.
We note that the absolute magnitude of the correlation coefficient is generally small because over 99.9\% of $y_{it}$ values are 0.
}
\label{fig:ref_corr}
\end{figure}

\paragraph{Evidence for reference-monotonicity.}
Next, we test the reference-monotonicity assumption that, under any discount $v_{it}$ offered today, purchase propensity is higher when the reference discount is smaller.
To do this, we group coupon values into \emph{small} discounts $(0.12,0.15)$ vs.\ \emph{large} discounts $(0.17,0.20)$.
For each case of $v_{it}$ being small or large, we consider different memory lengths $\ell\in\{3,4,5,7\}$, and report the percent change in purchase rate when the reference $\max\{v_{i,t-\ell},\ldots,v_{i,t-1}\}$ lies in the small set $(0.12,0.15)$ instead of the large set $(0.17,0.20)$. \Cref{tab:reference-monotonicity} shows that the change is positive in nearly all cases, providing support for the reference-monotonicity assumption.

\begin{table}[]
\centering
\begin{tabular}{c|cc}
\toprule
 $\ell$ &  $v_{it}\in\{0.12,0.15\}$ &  $v_{it}\in\{0.17,0.20\}$ \\
\midrule
 3 &        +7.1\% &        -1.1\%  \\
 4 &        +30.1\%  &         +16.5\%  \\
 5 &        +70.3\%  &         +55.9\%  \\
 7 &       +137.1\%  &        +12.6\%  \\
\bottomrule
\end{tabular}
\caption{Percent changes in purchase rate when the reference value $\max\{v_{i,t-\ell},\ldots,v_{i,t-1}\}$ lies in $\{0.12,0.15\}$ instead of $\{0.17,0.20\}$, reported separately for different memory lengths $\ell$ and different values of $v_{it}$.}
\label{tab:reference-monotonicity}
\end{table}

\section{Theoretical Model and Results} \label{sec:theory}

We provide the theoretical model and results outlined in \Cref{sec:introTheory}.
To be consistent with the literature on dynamic pricing with reference effects, just for this \namecref{sec:theory}, we switch the language from offering a sequence of \textit{discounts} (where bigger is better for the customer) to offering a sequence of \textit{prices} (where lower is better for the customer).  Therefore, the reference value is defined to be the \textit{minimum (instead of maximum)} of recent offerings.
We omit the index $i$ of the single customer.

\subsection{Model and MDP Formulation} \label{sec:model}

\paragraph{Model.}
There is a finite set of feasible prices $\cP$
from which an infinite sequence of prices $(p_t)_{t=1}^\infty$ is offered to the customer ($\cP=\{.80,.83,.85,.88,.90\}$ would correspond to the discount values $\cV=\{.10,.12,.15,.17,.20\}$ from our partner retailer).  At each time $t$, the \textit{reference} $r_t$ of the customer is $\min\{p_{t-\ell},\ldots,p_{t-1}\},$ the best price seen in the $\ell$ previous time steps (if $t\le\ell$, then $r_t$ is understood to be $\min\{p_1,\ldots,p_{t-1}$\}, with $r_1=\op$ where $\op:=\max\{p:p\in\cP\}$ denotes the highest price in $\cP$).
Here $\ell$ is a positive integer denoting the \textit{memory length} of the customer.

We note that both the price $p_t$ and the reference $r_t$ always lie in $\cP$.
We let $g(r,p)$ denote the immediate "gain", e.g.\ revenue, from offering price $p\in\cP$ to the customer when their reference is $r\in\cP$.
We allow for an arbitrary gain function $g$ as long as it satisfies the following assumption.

\begin{definition} \label{def:refMon}
A gain function $g:\cP^2\to\bR$ is said to be \textit{reference-monotone} if for any fixed $p\in\cP$, the gain $g(r,p)$ is weakly increasing in $r$.
\end{definition}

\Cref{def:refMon} is a mild assumption that is almost axiomatic in reference price models, where regardless of what price $p$ the seller plans on offering, they are better off if the reference price $r$ they are competing against is higher (and hence worse in the customer's mind).
Importantly, we do not make any assumptions on how $g(r,p)$ changes with $p$, which allows for different objectives such as market share (in which case $g(r,p)$ is generally decreasing in $p$) or revenue (in which case $g(r,p)$ is often increasing and then decreasing in $p$).

\paragraph{MDP formulation.}
The problem of maximizing long-run average gain
\begin{align} \label{eqn:lragObj}
\lim_{T\to\infty}\frac1T\sum_{t=1}^T g(r_t,p_t)
\end{align}
can be formulated as the following MDP.
The state space is $\cP^\ell$, with a state $s=(s^1,\ldots,s^\ell)$ denoting the $\ell$ previous prices seen by the customer, where $s^\ell$ is most recent.
(Although the reference only depends on the minimum of the $\ell$ previous prices, the evolution of this minimum depends on the exact sequence of $\ell$ previous prices, which must all be tracked in the state.)
The action space from any state is $\cP$, denoting the next price to offer.
When any action $p\in\cP$ is taken from any state $s\in\cP^\ell$, the reward is $g(\min\{s^1,\ldots,s^\ell\},p)$, and the state deterministically transitions to $(s^2,\ldots,s^{\ell},p)$.
The starting state is $s_1=(\op,\ldots,\op)$.

Standard theory for finite MDP's \citep{puterman2014markov} confirms that~\eqref{eqn:lragObj} can be maximized using a policy which is stationary and deterministic, defined by a mapping $\pi:\cP^\ell\to\cP$ specifying the price to offer from any state.
However, optimizing over such policies is still difficult, because our MDP is exponential-sized.  Therefore, we consider the following reformulation instead.

\paragraph{Reduction to optimization over price cycles.}
Because the transitions are deterministic in our MDP, a stationary deterministic policy $\pi$ induces a periodic steady state, where for some cycle length $c\le|\cP|^\ell$, there exists a time $T_0$ after which the states visited will cycle between $s(0),\ldots,s(c-1)$.
The prices offered would cycle between $\pi(s(0)),\ldots,\pi(s(c-1))$, which we denote using $\pi_0,\ldots,\pi_{c-1}$ respectively.
Given $\pi_0,\ldots,\pi_{c-1}$, the cycle of states can be reconstructed to be $s(t)=(\pi_{t-\ell\mod c},\ldots,\pi_{t-1\mod c})$ for all $t=0,\ldots,c-1$ (noting that $c<\ell$ is possible), and hence the objective~\eqref{eqn:lragObj} (irrespective of starting state) equals
\begin{align} \label{eqn:cycleObj}
\frac 1c \sum_{t=0}^{c-1} g\Big(\min\{\pi_{t-\ell\mod c},\ldots,\pi_{t-1\mod c}\},\pi_t\Big).
\end{align}

\subsection{Characterization of Optimal Policies}

Maximizing objective~\eqref{eqn:cycleObj} over positive integers $c$ and prices cycles $(\pi_0,\ldots,\pi_{c-1})\in\cP^c$ is still computationally challenging, because a priori, $c$ could be as long as $|\cP|^\ell$, with the same price appearing multiple times in different contexts.  For an arbitrary gain function $g:\cP^2\to\bR$, this could indeed be the case (see \Cref{ex:withoutRefMon} below).  
However, under the reference-monotonicity assumption, our main result drastically reduces the search space, to "$\ell$-up-1-down" price cycles.

\begin{definition}[Notation]
We hereafter represent prices cycles using strings over the finite alphabet $\cP$, and let $p^k$ denote the concatenation of $k$ copies of  token $p$, for any positive integer $k$ and $p\in\cP$.
For example, if $\cP=\{1,2,3\}$, then $12^3$ denotes the string 1222, which is the price cycle with $c=4$ and $(\pi_0,\ldots,\pi_3)=(1,2,2,2)$.
We note that this cycle can be equivalently represented using the strings 2221, 2212, or 12221222.
\end{definition}

\begin{definition} \label{def:luid}
A price cycle is said to be \textit{\luid} if it can be represented by a string of the form $\rho_0^{k_0}\cdots\rho_{d-1}^{k_{d-1}}$, where $d\le|\cP|$ is a positive integer, $\rho_0,\ldots,\rho_{d-1}$ are \textit{distinct} prices in $\cP$, and
\begin{align*}
k_t &=
\begin{cases}
\ell, &\rho_t>\rho_{t-1\mod d} \\
1, &\rho_t<\rho_{t-1\mod d}
\end{cases}
&\forall t=0,\ldots,d-1.
\end{align*}
We call $(\rho_0,\ldots,\rho_{d-1})$ the \textit{generator cycle}, with length $d$.
The price cycle itself has length $c=\sum_{t=0}^{d-1}(1+(\ell-1)\bI(\rho_t>\rho_{t-1\mod d}))$.
\end{definition}

\begin{example}[Examples of \luid Price Cycles]
Let $\cP=\{1,2,3\}$ and $\ell=3$.  Then 1222, 1222333, and 13332 represent \luid price cycles, induced by generator cycles $(\rho_0,\rho_1)=(1,2)$, $(\rho_0,\rho_1,\rho_2)=(1,2,3)$, and $(\rho_0,\rho_1,\rho_2)=(1,3,2)$ respectively.
Meanwhile, 12 does not represent an \luid price cycle because the "2" does not have $\ell$ copies despite being greater than 1.
Also, 111222 does not represent an \luid cycle because the "1" has multiple copies despite being less than 2.
Finally, 12223332 has the correct number of copies but does not represent an \luid cycle because copies of "2" occur on two different occasions in the cycle.
\end{example}

\begin{theorem} \label{thm:theoryMain}
For any finite $\cP$, memory length $\ell$, and reference-monotone gain function $g$, there exists an \luid price cycle that maximizes~\eqref{eqn:cycleObj}.
\end{theorem}

\Cref{thm:theoryMain} shows that the optimization problem can be restricted to \luid price cycles, and we will use this result in \Cref{sec:theoryLooseEnds} to construct polynomial-time algorithms.  We first provide some intuition for \Cref{thm:theoryMain} by showing that without the reference-monotonicity assumption, the unique optimal price cycle can be quite complex.

\begin{example}[Necessity of Reference-monotonicity] \label{ex:withoutRefMon}
Let $\ell=2$, $\cP=\{1,2,3,4\}$, and $g(r,p)$ be defined as in \Cref{tab:strangeEg}.  The optimal price cycle for this example is 414243, whose objective~\eqref{eqn:cycleObj} equals
\begin{align*}
\frac16 \big(g(3,4)+g(3,1)+g(1,4)+g(1,2)+g(2,4)+g(2,3)\big)=1.
\end{align*}
It is easy to check that no other price cycle (\luid or not) can have an average gain of 1.
\end{example}

\begin{table}[]
\centering
\begin{tabular}{c|cccc}
$g(r,p)$ & $p=1$ & $p=2$ & $p=3$ & $p=4$ \\
\hline
$r=1$ & 0 & 1 & 0 & 1 \\
$r=2$ & 0 & 0 & 1 & 1 \\
$r=3$ & 1 & 0 & 0 & 1 \\
$r=4$ & 0 & 0 & 0 & 0 \\
\end{tabular}
\caption{The gain function $g$ used for \Cref{ex:withoutRefMon}.  A reference-monotone $g$ would be weakly increasing down each column; this gain function is not.}
\label{tab:strangeEg}
\end{table}

\subsection{Proof of \Cref{thm:theoryMain}}

Our goal is to prove that for any given $c$ and price cycle $\pi_0,\ldots,\pi_{c-1}$, its objective value~\eqref{eqn:cycleObj} can be upper-bounded by that of an \luid price cycle.
Although our proof is non-constructive, the result of \Cref{thm:theoryMain} itself can be used formulate a reduced MDP that allows for polynomial-time algorithms, as we show in \Cref{sec:theoryLooseEnds}.

Our proof proceeds in two steps.
\Cref{lem:reduction} first reduces the given price cycle to an intermediate form, and the proof then repeatedly applies \Cref{lem:reset} until we end up with an \luid price cycle, whose objective value has not decreased.
To aid the reader, we provide the following example illustrating the two steps of the reduction, that can be referenced while reading the proofs.

\begin{example}
Let $\cP=\{1,2,3,4,5,6\}$, $\ell=3$, and suppose we are given the price cycle below.  Its \textit{low points} (as defined in the proof of \Cref{lem:reduction}) are underlined.
$$
\underline{1}2345\underline{3}4\underline{2}65345\underline{3}
$$
The substrings between the low points are replaced in the following way:
\begin{itemize}
\item Substring 2345 between the first two low points is replaced with either 222555, 333, or 444;
\item Substring 4 is replaced with the empty string, or it is concluded that the objective value of the price cycle is worse than that of constantly offering price 4;
\item Substring 65345 is replaced with either 666444, 555555, or 333;
\item The empty substring between the final \underline{3} that wraps around to the initial \underline{1} is unchanged.
\end{itemize}
\Cref{lem:reduction} shows that at least one combination of these replacements would not decrease the objective value.
Suppose for illustration this replacement is
$$
\textbf{\underline{1}}22\textbf{2}55\textbf{5\underline{3}\underline{2}}66\textbf{6444\underline{3}}
$$
where the original low points are still underlined and now we have also bolded the \textit{reset points} (as defined in \Cref{lem:reset}).

The second part of the proof finds two distinct reset points with the same number, and breaks the cycle into two starting from these points, wrapping around as necessary.  For example, taking the 3's as the reset points, we can consider the two shorter cycles below:
$$
\textbf{\underline{3}\underline{2}}66\textbf{6444}; \qquad \textbf{\underline{3}}\textbf{\underline{1}}22\textbf{2}55\textbf{5}.
$$
By \Cref{lem:reset}, the better of these cycles has objective value no less than that of the original cycle.
We take the better cycle and repeat this process until there are no distinct reset points with the same number.  For example, if the better cycle was 32666444, then we might end up with the cycle 326664, or just 4 (both of which are \luid price cycles).  On the other hand, if the better cycle was 31222555, then there is no further reduction.
In either case, the process must terminate with an \luid price cycle. \qed
\end{example}

We now proceed with the formal proof.

\begin{lemma} \label{lem:reduction}
For any price cycle $\pi_0,\ldots,\pi_{c-1}$, its objective value~\eqref{eqn:cycleObj} is upper-bounded by that of a price cycle represented by a string of the form $\rho_0^{k_0}\cdots\rho_{d-1}^{k_{d-1}}$,
where $d$ is a positive integer, and for all $t=0,\ldots,d-1$ we have $\rho_t\in\cP$, $k_t\in\{1,\ell\}$, and $k_t=1$ implying $\rho_t\le\rho_{t-1\mod d}$.
\end{lemma}

\begin{proof}[Proof of \Cref{lem:reduction}]
In the initial price cycle $\pi_0,\ldots,\pi_{c-1}$, define a time $t\in\{0,\ldots,c-1\}$ to be a \textit{low point} if $\pi_t\le\min\{\pi_{t-\ell\mod c},\ldots,\pi_{t-1\mod c}\}$.  That is, $t$ is a low point if it is at most the reference price at time $t$.
Let $t_1,\ldots,t_m$ denote the low points in the initial cycle, with $0\le t_1<\cdots<t_m\le c-1$, noting that there must be at least one low point.

We consider the substrings of prices between low points: $$\pi_{t_m+1\mod c}\cdots\pi_{t_1-1\mod c};\quad \pi_{t_1+1\mod c}\cdots\pi_{t_2-1\mod c};\quad \ldots \quad;\quad\pi_{t_{m-1}+1\mod c}\cdots\pi_{t_m-1\mod c}.$$  Note that there are $m$ such substrings, some of which may be empty.
We iteratively replace these substrings with new substrings in which every token appears $\ell$ consecutive times, showing that the objective value~\eqref{eqn:cycleObj} does not go down, and that all low points are preserved.
This would complete the proof because after the $m$ replacements, all tokens would either appear $\ell$ consecutive times, or be a low point which implies that it is no greater than the immediately preceding price.

To ease notation, at each iteration we re-index (rotate) the cycle so that the current substring under consideration is $\pi_0\cdots\pi_{t'-1}$, preceded by a low point at time $c-1$ and succeeded by a low point at time $t'$.  None of $0,\ldots,t'-1$ being low points guarantees that $\pi_t>\pi_{c-1}$ for all $t=0,\ldots,t'-1$, because otherwise the smallest $t\in\{0,\ldots,t'-1\}$ for which $\pi_t\le\pi_{c-1}$ would be a low point.

\paragraph{Case 1: $t'<\ell$.}
Because $t'<\ell$, the reference price for each time $t=0,\ldots,t'-1$ is at most (in fact equal to) $\pi_{c-1}$.  Therefore, we can express the objective value of the current price cycle as
\begin{align*}
&\frac{1}{c}\left(\sum_{t=0}^{t'-1} g(\pi_{c-1},\pi_t) + \sum_{t=t'}^{c-1}g(\min\{\pi_{t-\ell\mod c},\ldots,\pi_{t-1\mod c}\},\pi_t)\right)
\\ &\le\frac{\sum_{t=0}^{t'-1} g(\pi_t,\pi_t) + \sum_{t=t'}^{c-1}g(\min\{\pi_{t-\ell\mod c},\ldots,\pi_{t-1\mod c}\},\pi_t)}{t'+(c-t')}
\\ &\le\max\left\{g(\pi_0,\pi_0),\ldots,g(\pi_{t'-1},\pi_{t'-1}),\frac{\sum_{t=t'}^{c-1}g(\min\{\pi_{t-\ell\mod c},\ldots,\pi_{t-1\mod c}\},\pi_t)}{c-t'}\right\}
\end{align*}
where the first inequality follows from reference-monotonicity because $\pi_t>\pi_{c-1}$ for all $t=0,\ldots,t'-1$, and the second inequality follows elementarily\footnote{For real numbers $a_1,\ldots,a_n$ and $b_1,\ldots,b_n>0$, it holds that $(a_1+\cdots+a_n)/(b_1+\cdots+b_n)\le \max\{\frac{a_1}{b_1},\ldots,\frac{a_n}{b_n}\}$.\label{ftnt:elem}}.
If the $\max$ in the final expression is attained at any of the first $t'$ arguments, then \Cref{lem:reduction} would be immediately proven, because we have upper-bounded the objective value by a price cycle that is the constant price $\pi_t$ for some $t\in\{0,\ldots,t'-1\}$.

Otherwise, if the $\max$ in the final expression is attained at the final argument, then we claim this is equal to the objective value of the price cycle $\pi_{t'}\pi_{t'+1}\cdots\pi_{c-1}$ (formed by replacing $\pi_0\cdots\pi_{t'-1}$ with the empty substring).  To see this, note that the reference price at any time $t\ge t'+\ell$ is unchanged.  For time steps $t=t',\ldots,t'+\ell-1$, the reference price is still $\pi_{t'}$ because $\pi_{t'}\le\pi_{c-1}$ by virtue of $t'<\ell$ and $t'$ being a low point, and $\pi_{c-1}\le\min\{\pi_{c-1-\ell\mod c},\ldots,\pi_{c-2\mod c}\}$ by virtue of $c$ being a low point (which can be checked to hold under the new cycle even if $\pi_{c-1-\ell\mod c}$ wraps around).  That is, both $t'$ and $c-1$ are still low points in the new price cycle.
This completes the proof of Case 1.

\paragraph{Case 2: $t'\ge\ell$.}
We consider replacing $\pi_0\cdots\pi_{t'-1}$ with a substring of the form
\begin{align} \label{eqn:substrings}
\pi_j^\ell \pi_{\ell+j}^\ell\pi_{2\ell+j}^\ell\cdots\pi_{\lfloor (t'-1-j)/\ell\rfloor\ell+j}^\ell, && j\in\{0,\ldots,\ell-1\}.
\end{align}
For each substring $j=0,\ldots,\ell-1$, note that it has length $(1+\lfloor (t'-1-j)/\ell\rfloor)\ell$.
The reference price for the $\ell$ copies of its first price $\pi_j$ is $\pi_{c-1}=\min\{\pi_{c-1},\pi_j\}$, because $c-1$ is a low point and $\pi_j>\pi_{c-1}$.
Meanwhile, for all $j'=1,\ldots,\lfloor (t'-1-j)/\ell\rfloor$, the reference price for the $\ell$ copies of its price $\pi_{j'\ell+j}$ is at least $\min\{\pi_{(j'-1)\ell+j},\pi_{j'\ell+j}\}$.
Finally, for any replacement substring $j$ the reference prices at all times $t\ge t'$ remain unchanged, because $\pi_{t'}\le\pi_{\lfloor (t'-1-j)/\ell\rfloor\ell+j}$ by virtue of $t'$ being a low point and $\pi_{\lfloor (t'-1-j)/\ell\rfloor\ell+j}$ existing (where $\lfloor (t'-1-j)/\ell\rfloor\ge0$ due to $t'\ge\ell$).
Therefore, the objective value under each replacement substring $j=0,\ldots,\ell-1$ is at least
\begin{align*}
\frac{\ell g(\pi_{c-1},\pi_j)
+\ell\sum_{j'=1}^{\lfloor (t'-1-j)/\ell\rfloor} g(\min\{\pi_{(j'-1)\ell+j},\pi_{j'\ell+j}\},\pi_{j'\ell+j})+\sum_{t=t'}^{c-1}g(\min\{\pi_{t-\ell},\ldots,\pi_{t-1}\})}{(1+\lfloor (t'-1-j)/\ell\rfloor)\ell+c-t'}.
\end{align*}

Now, the objective value of the current price cycle with substring $\pi_0\cdots\pi_{t'-1}$ can be bounded:
\begin{align*}
&\frac{1}{c}\left(\sum_{t=0}^{\ell-1} g(\pi_{c-1},\pi_t)
+\sum_{t=\ell}^{t'-1} g(\min\{\pi_{t-\ell},\ldots,\pi_{t-1}\},\pi_t)
+\sum_{t=t'}^{c-1}g(\min\{\pi_{t-\ell},\ldots,\pi_{t-1}\},\pi_t)\right)
\\ &\le\frac{1}{c}\left(\sum_{t=0}^{\ell-1} g(\pi_{c-1},\pi_t)
+\sum_{t=\ell}^{t'-1} g(\min\{\pi_{t-\ell},\pi_t\},\pi_t)
+\sum_{t=t'}^{c-1}g(\min\{\pi_{t-\ell},\ldots,\pi_{t-1}\},\pi_t)\right)
\\ &=\frac{1}{c}\left(\sum_{j=0}^{\ell-1} \left(g(\pi_{c-1},\pi_j)
+\sum_{j'=1}^{\lfloor (t'-1-j)/\ell\rfloor} g(\min\{\pi_{(j'-1)\ell+j},\pi_{j'\ell+j}\},\pi_{j'\ell+j})\right)
+\sum_{t=t'}^{c-1}g(\min\{\pi_{t-\ell},\ldots,\pi_{t-1}\},\pi_t)\right)
\\ &=\frac{1}{c}\left(\frac1\ell\sum_{j=0}^{\ell-1} \left(\ell g(\pi_{c-1},\pi_j)
+\ell\sum_{j'=1}^{\lfloor (t'-1-j)/\ell\rfloor} g(\min\{\pi_{(j'-1)\ell+j},\pi_{j'\ell+j}\},\pi_{j'\ell+j})
+\sum_{t=t'}^{c-1}g(\min\{\pi_{t-\ell},\ldots,\pi_{t-1}\},\pi_t)\right)\right)
\\ &=\frac{\sum_{j=0}^{\ell-1} \left(\ell g(\pi_{c-1},\pi_j)
+\ell\sum_{j'=1}^{\lfloor (t'-1-j)/\ell\rfloor} g(\min\{\pi_{(j'-1)\ell+j},\pi_{j'\ell+j}\},\pi_{j'\ell+j})
+\sum_{t=t'}^{c-1}g(\min\{\pi_{t-\ell},\ldots,\pi_{t-1}\},\pi_t)\right)}{\ell(c-t')+\ell\sum_{j=0}^{\ell-1}(1+\lfloor (t'-1-j)/\ell\rfloor)}
\\ &\le\max_{j=0,\ldots,\ell-1}\frac{\ell g(\pi_{c-1},\pi_j)
+\ell\displaystyle\sum_{j'=1}^{\lfloor (t'-1-j)/\ell\rfloor} g(\min\{\pi_{(j'-1)\ell+j},\pi_{j'\ell+j}\},\pi_{j'\ell+j})
+\sum_{t=t'}^{c-1}g(\min\{\pi_{t-\ell},\ldots,\pi_{t-1}\},\pi_t)}{1+(\lfloor (t'-1-j)/\ell\rfloor)\ell+c-t'}.
\end{align*}
The first inequality uses the fact that $\pi_t>\min\{\pi_{t-\ell},\ldots,\pi_{t-1}\}$ because $t$ is not a low point, and hence $g(\min\{\pi_{t-\ell},\ldots,\pi_{t-1}\},\pi_t)=g(\min\{\pi_{t-\ell},\ldots,\pi_{t}\},\pi_t)\le g(\min\{\pi_{t-\ell},\pi_{t}\},\pi_t)$ (with the inequality applying reference-monotonicity).
The first equality partitions the integers $0,\ldots,t'-1$ based on their remainder $j$ when divided by $\ell$, noting that $\lfloor (t'-1-j)/\ell\rfloor\ell+j$ is the largest integer less than $t'$ with a remainder of $j$ when divided by $\ell$.
The third equality holds because $\sum_{j=0}^{\ell-1}(1+\lfloor (t'-1-j)/\ell\rfloor)=t'$, which is easiest to see from the fact that there are $1+\lfloor (t'-1-j)/\ell\rfloor$ integers in $\{0,\ldots,t'-1\}$ with remainder $j$ when divided by $\ell$.
The final inequality follows from \Cref{ftnt:elem}.

Therefore, we can always replace $\pi_0\cdots\pi_{t'-1}$ with one of the substrings from~\eqref{eqn:substrings} without decreasing the objective value.  Moreover, $t'$ remains a low point because it is now preceded by $\ell$ copies of $\lfloor (t'-1-j)/\ell\rfloor\ell+j$, which appeared in $\{\pi_{t'-\ell},\ldots,\pi_{t'-1}\}$ and hence is at least $\pi_{t'}$.  This completes the proof of Case 2 and the overall proof of \Cref{lem:reduction}.
\end{proof}

\begin{lemma} \label{lem:reset}
In a price cycle $\pi_0,\ldots,\pi_{c-1}$, define a time $t\in\{0,\ldots,c-1\}$ to be a \textit{reset point} if $\pi_t=\min\{\pi_{t-\ell+1\mod c},\ldots,\pi_{t\mod c}\}$.  Suppose the price cycle $\pi_0,\ldots,\pi_{c-1}$ contains two distinct reset points with the same price, relabeled to be $t'-1$ and $c-1$ (i.e., $t'-1,c-1$ are both reset points with $t'<c$ and $\pi_{t'-1}=\pi_{c-1}$).  Then the objective value~\eqref{eqn:cycleObj} of the price cycle must be upper-bounded by that of either the price cycle $\pi_0,\ldots,\pi_{t'-1}$ or the price cycle $\pi_{t'},\ldots,\pi_{c-1}$.
\end{lemma}

\begin{proof}[Proof of \Cref{lem:reset}]
The objective value of the original price cycle $\pi_0,\ldots,\pi_{c-1}$ equals
\begin{align*}
\frac 1{t'+(c-t')} \Bigg(&\sum_{t=0}^{\min\{t'-1,\ell-1\}}g(\min\{\pi_{c-1},\pi_0,\ldots,\pi_{t-1}\},\pi_t)
+ \sum_{t=\min\{t'-1,\ell-1\}+1}^{t'-1}g(\min\{\pi_{t-\ell},\ldots,\pi_{t-1}\},\pi_t)
\\ &+ \sum_{t=t'}^{\min\{c-1,\ell-1\}}g(\min\{\pi_{t'-1},\pi_0,\ldots,\pi_{t-1}\},\pi_t)
+ \sum_{t=\min\{c-1,\ell-1\}+1}^{c-1}g(\min\{\pi_{t-\ell},\ldots,\pi_{t-1}\},\pi_t)\Bigg)
\end{align*}
where we have used the fact that both $t'-1$ and $c-1$ are reset points.
Using the same fact, the objective value of price cycle $\pi_0,\ldots,\pi_{t'-1}$ equals
$$
\frac1{t'}\left(\sum_{t=0}^{\min\{t'-1,\ell-1\}}g(\min\{\pi_{t'-1},\pi_0,\ldots,\pi_{t-1}\},\pi_t)
+ \sum_{t=\min\{t'-1,\ell-1\}+1}^{t'-1}g(\min\{\pi_{t-\ell},\ldots,\pi_{t-1}\},\pi_t)\right)
$$
while the objective value of price cycle $\pi_{t'},\ldots,\pi_{c-1}$ equals
$$
\frac1{c-t'}\left(\sum_{t=t'}^{\min\{c-1,\ell-1\}}g(\min\{\pi_{c-1},\pi_0,\ldots,\pi_{t-1}\},\pi_t)
+ \sum_{t=\min\{c-1,\ell-1\}+1}^{c-1}g(\min\{\pi_{t-\ell},\ldots,\pi_{t-1}\},\pi_t)\right).
$$
Because $\pi_{t'-1}=\pi_{c-1}$, it follows
from \Cref{ftnt:elem}
that the objective value of the original price cycle is upper-bounded by the maximum of the latter two expressions.
\end{proof}

\begin{proof}[Completing the proof of \Cref{thm:theoryMain}]
Consider the upper-bounding price cycle that is the result of the reduction in \Cref{lem:reduction}.
For $t=0,\ldots,d-1$, first suppose $k_t=1$.  Then it is guaranteed that $\rho_t\le\rho_{t-1\mod d}$, which means that $\rho_t$ corresponds to a reset point.  Indeed, this is immediate if $k_{t-1\mod d}=\ell$, and if $k_{t-1\mod d}=1$ then $\rho_{t-1\mod d}\le\rho_{t-2\mod d}$ so we can iteratively apply the same argument.  On the other hand, now suppose $k_t=\ell$.  In this case, if $\rho_t\le\rho_{t-1\mod d}$, then the $\ell$ copies of $\rho_t$ all correspond to reset points, by the same argument as before; if otherwise $\rho_t>\rho_{t-1\mod d}$, then only the final copy of $\rho_t$ corresponds to a reset point.

All in all, we have proven that every $t=0,\ldots,d-1$ corresponds to at least 1 reset point, and corresponds to $\ell$ reset points if $k_t=\ell$ and $\rho_t\le\rho_{t-1\mod d}$.
We now iteratively apply \Cref{lem:reset} to reduce the price cycle whenever distinct reset points have the same price, noting that after each reduction the upper-bounding price cycle still takes the form described in \Cref{lem:reduction} and the same properties still hold.
Let $\rho_0^{k_0}\cdots\rho_{d-1}^{k_{d-1}}$ represent the final price cycle which no longer has distinct reset points with the same price.
It must be the case that $\rho_0,\ldots,\rho_{d-1}$ are distinct prices in $\cP$, and moreover, if $k_t=\ell$ then we must have $\rho_t>\rho_{t-1\mod d}$, because otherwise all $\ell$ copies of $\rho_t$ would be reset points.  (In the case where $\ell=1$, this argument is irrelevant.)  The original property from \Cref{lem:reduction} that if $k_t=1$ then $\rho_t\le\rho_{t-1\mod d}$ also still holds; in fact we would have $\rho_t<\rho_{t-1\mod d}$, by distinctness.  This shows that the the final price cycle must satisfy precisely the definition of \luid, completing the proof of \Cref{thm:theoryMain}.
\end{proof}

\subsection{Computational Consequences; Tightness of \luid Characterization} \label{sec:theoryLooseEnds}

\Cref{thm:theoryMain} shows that to maximize long-run average gain~\eqref{eqn:cycleObj} over all positive integers $c\le|\cP|^\ell$ and price cycles $(\pi_0,\ldots,\pi_{c-1})\in\cP^c$, it suffices to search over \luid price cycles, which are defined by a generator cycle of at most $d$ distinct prices in $\cP$ (see \Cref{def:luid}).  We now show that the latter problem can be formulated using a reduced MDP.

For any generator cycle $\rho_0,\ldots,\rho_{d-1}$ of distinct prices in $\cP$, its long-run average gain~\eqref{eqn:cycleObj} equals
\begin{align} \label{eqn:lragOfGenerator}
\frac{\sum_{t=0}^{d-1}g(\rho_{t-1\mod d},\rho_t)(1+(\ell-1)\bI(\rho_t>\rho_{t-1\mod d}))}{\sum_{t=0}^{d-1}(1+(\ell-1)\bI(\rho_t>\rho_{t-1\mod d}))},
\end{align}
because the reference price is always $\rho_{t-1\mod d}$ while offering $\rho_t$, for all $t=0,\ldots,d-1$, due to the \luid structure that the previous price $\rho_{t-1\mod d}$ is repeated $\ell$ times if it is higher than the price before it.

We now construct an MDP where both the state and action space is $\cP$.  When any action $p\in\cP$ is taken from any state $r\in\cP$, the next state is always $p$;  the reward is $g(r,p)$ if $r\ge p$, and $\ell g(r,p)$ if $r<p$ but the \textit{transition takes $\ell$ steps} (instead of 1 step, so the reward-per-step is still $g(r,p)$).
The objective is to maximize long-run average reward per step.
We again restrict without loss to stationary deterministic policies\footnote{Technically this is a Semi-Markov Decision Process due to the inhomogeneous transition times, but the optimality of stationary deterministic policies still holds.  We can alternatively have homogeneous transition times if for all transitions from states $r$ to $p$ with $r<p$, we add $\ell-1$ dummy states in the middle with a single action for proceeding.}, to see that the optimal policy is defined by a cycle $\rho_0,\ldots,\rho_{d-1}$ of distinct states in $\cP$.
Moreover, the long-run average reward per step of this policy equals exactly~\eqref{eqn:lragOfGenerator}.

Therefore, the optimization problem over \luid price cycles is captured by this reduced MDP, which can be solved via the following system of infinite-horizon undiscounted Bellman's equations:
\begin{align} \label{eqn:bellman}
h(r) &=\max_{p\in\cP}\Big(\big(g(r,p)-\OPT\big)\big(1+(\ell-1)\bI(r<p)\big)+h(p)\Big) &\forall r\in\cP.
\end{align}
Here, variable $\OPT$ denotes the optimal long-run average reward, and $h(p)$ denotes the bias for each state $p$, one of which can be normalized to 0 so that there are both $|\cP|$ variables and equations.  The multiplication by the factor of $1+(\ell-1)\bI(r<p)$ accounts for the transition times.  For further details and algorithms that solve this in polynomial-time, we defer to \citet{puterman2014markov}.

We can use~\eqref{eqn:bellman} to also prove that our characterization of \luid policies is tight.

\begin{theorem}[proven in \Cref{pf:tightness}] \label{thm:tightness}
Let $\rho_0,\ldots,\rho_{d-1}$ be any cycle of distinct prices in $\cP$.
Under any memory length $\ell$, there exists a reference-monotone gain function $g:\cP^2\to\bR$ for which the unique optimal price cycle is $\rho_0^{k(\rho_{d-1},\rho_0)}\rho_1^{k(\rho_{0},\rho_1)}\cdots\rho_{d-1}^{k(\rho_{d-2},\rho_{d-1})}$, where $k(r,p):=1+(\ell-1)\bI(r<p)$.    
\end{theorem}

\section{Conclusion and Post-mortem}

Personalizing promotions and curating user journeys is more possible than ever, yet optimizing these for the long-term is a perennially unsolved business problem.
Our paper combines experience from a successful real-world deployment of personalized promotions with a theoretical model and structural result for optimizing personalized promotion cycles for the long-term.
Unfortunately, our partner de-prioritized personalizing promotions and machine learning more broadly at the end of 2024, in favor of offering the biggest discount to all customers every day, removing all personalization and novelty of big discounts.
This prevented us from testing a more sophisticated \luid policy, and even our simple algorithm was phased out of deployment by early 2025.

%
%
%
%
%

\clearpage
\bibliographystyle{ACM-Reference-Format}
\bibliography{bibliography}


\begin{thebibliography}{31}


\ifx \showCODEN    \undefined \def \showCODEN     #1{\unskip}     \fi
\ifx \showDOI      \undefined \def \showDOI       #1{#1}\fi
\ifx \showISBNx    \undefined \def \showISBNx     #1{\unskip}     \fi
\ifx \showISBNxiii \undefined \def \showISBNxiii  #1{\unskip}     \fi
\ifx \showISSN     \undefined \def \showISSN      #1{\unskip}     \fi
\ifx \showLCCN     \undefined \def \showLCCN      #1{\unskip}     \fi
\ifx \shownote     \undefined \def \shownote      #1{#1}          \fi
\ifx \showarticletitle \undefined \def \showarticletitle #1{#1}   \fi
\ifx \showURL      \undefined \def \showURL       {\relax}        \fi
\providecommand\bibfield[2]{#2}
\providecommand\bibinfo[2]{#2}
\providecommand\natexlab[1]{#1}
\providecommand\showeprint[2][]{arXiv:#2}

\bibitem[Agrawal and Tang(2024)]%
        {agrawal2024dynamic}
\bibfield{author}{\bibinfo{person}{Shipra Agrawal} {and} \bibinfo{person}{Wei Tang}.} \bibinfo{year}{2024}\natexlab{}.
\newblock \showarticletitle{Dynamic Pricing and Learning with Long-term Reference Effects}. In \bibinfo{booktitle}{\emph{Proceedings of the 25th ACM Conference on Economics and Computation}}. \bibinfo{pages}{72--72}.
\newblock


\bibitem[Albert and Goldenberg(2022)]%
        {albert2022commerce}
\bibfield{author}{\bibinfo{person}{Javier Albert} {and} \bibinfo{person}{Dmitri Goldenberg}.} \bibinfo{year}{2022}\natexlab{}.
\newblock \showarticletitle{E-commerce promotions personalization via online multiple-choice knapsack with uplift modeling}. In \bibinfo{booktitle}{\emph{Proceedings of the 31st ACM International Conference on Information \& Knowledge Management}}. \bibinfo{pages}{2863--2872}.
\newblock


\bibitem[Baek et~al\mbox{.}(2025)]%
        {baek2025policy}
\bibfield{author}{\bibinfo{person}{Jackie Baek}, \bibinfo{person}{Justin~J Boutilier}, \bibinfo{person}{Vivek~F Farias}, \bibinfo{person}{Jonas~Oddur Jonasson}, {and} \bibinfo{person}{Erez Yoeli}.} \bibinfo{year}{2025}\natexlab{}.
\newblock \showarticletitle{Policy optimization for personalized interventions in behavioral health}.
\newblock \bibinfo{journal}{\emph{Manufacturing \& Service Operations Management}} \bibinfo{volume}{27}, \bibinfo{number}{3} (\bibinfo{year}{2025}), \bibinfo{pages}{770--788}.
\newblock


\bibitem[Besbes and Lobel(2015)]%
        {besbes2015intertemporal}
\bibfield{author}{\bibinfo{person}{Omar Besbes} {and} \bibinfo{person}{Ilan Lobel}.} \bibinfo{year}{2015}\natexlab{}.
\newblock \showarticletitle{Intertemporal price discrimination: Structure and computation of optimal policies}.
\newblock \bibinfo{journal}{\emph{Management Science}} \bibinfo{volume}{61}, \bibinfo{number}{1} (\bibinfo{year}{2015}), \bibinfo{pages}{92--110}.
\newblock


\bibitem[Calmon et~al\mbox{.}(2021)]%
        {calmon2021revenue}
\bibfield{author}{\bibinfo{person}{Andre~P Calmon}, \bibinfo{person}{Florin~D Ciocan}, {and} \bibinfo{person}{Gonzalo Romero}.} \bibinfo{year}{2021}\natexlab{}.
\newblock \showarticletitle{Revenue management with repeated customer interactions}.
\newblock \bibinfo{journal}{\emph{Management Science}} \bibinfo{volume}{67}, \bibinfo{number}{5} (\bibinfo{year}{2021}), \bibinfo{pages}{2944--2963}.
\newblock


\bibitem[Chang et~al\mbox{.}(2024)]%
        {chang2024pricing}
\bibfield{author}{\bibinfo{person}{Jiacheng Chang}, \bibinfo{person}{Xiao Lei}, {and} \bibinfo{person}{Feng Tian}.} \bibinfo{year}{2024}\natexlab{}.
\newblock \showarticletitle{Pricing and Addiction Control for Digital Services}.
\newblock \bibinfo{journal}{\emph{Available at SSRN 4962550}} (\bibinfo{year}{2024}).
\newblock


\bibitem[Chen et~al\mbox{.}(2017)]%
        {chen2017efficient}
\bibfield{author}{\bibinfo{person}{Xin Chen}, \bibinfo{person}{Peng Hu}, {and} \bibinfo{person}{Zhenyu Hu}.} \bibinfo{year}{2017}\natexlab{}.
\newblock \showarticletitle{Efficient algorithms for the dynamic pricing problem with reference price effect}.
\newblock \bibinfo{journal}{\emph{Management Science}} \bibinfo{volume}{63}, \bibinfo{number}{12} (\bibinfo{year}{2017}), \bibinfo{pages}{4389--4408}.
\newblock


\bibitem[Cohen et~al\mbox{.}(2020)]%
        {cohen2020efficient}
\bibfield{author}{\bibinfo{person}{Maxime~C Cohen}, \bibinfo{person}{Swati Gupta}, \bibinfo{person}{Jeremy~J Kalas}, {and} \bibinfo{person}{Georgia Perakis}.} \bibinfo{year}{2020}\natexlab{}.
\newblock \showarticletitle{An efficient algorithm for dynamic pricing using a graphical representation}.
\newblock \bibinfo{journal}{\emph{Production and Operations Management}} \bibinfo{volume}{29}, \bibinfo{number}{10} (\bibinfo{year}{2020}), \bibinfo{pages}{2326--2349}.
\newblock


\bibitem[Cohen-Hillel et~al\mbox{.}(2023)]%
        {cohen2023high}
\bibfield{author}{\bibinfo{person}{Tamar Cohen-Hillel}, \bibinfo{person}{Kiran Panchamgam}, {and} \bibinfo{person}{Georgia Perakis}.} \bibinfo{year}{2023}\natexlab{}.
\newblock \showarticletitle{High-low promotion policies for peak-end demand models}.
\newblock \bibinfo{journal}{\emph{Management Science}} \bibinfo{volume}{69}, \bibinfo{number}{4} (\bibinfo{year}{2023}), \bibinfo{pages}{2016--2050}.
\newblock


\bibitem[Dai et~al\mbox{.}(2024)]%
        {dai2024data}
\bibfield{author}{\bibinfo{person}{Jinglong Dai}, \bibinfo{person}{Hanwei Li}, \bibinfo{person}{Weiming Zhu}, \bibinfo{person}{Jianfeng Lin}, {and} \bibinfo{person}{Binqiang Huang}.} \bibinfo{year}{2024}\natexlab{}.
\newblock \showarticletitle{Data-Driven Real-time Coupon Allocation in the Online Platform}.
\newblock \bibinfo{journal}{\emph{arXiv preprint arXiv:2406.05987}} (\bibinfo{year}{2024}).
\newblock


\bibitem[den Boer and Keskin(2022)]%
        {den2022dynamic}
\bibfield{author}{\bibinfo{person}{Arnoud~V den Boer} {and} \bibinfo{person}{N~Bora Keskin}.} \bibinfo{year}{2022}\natexlab{}.
\newblock \showarticletitle{Dynamic pricing with demand learning and reference effects}.
\newblock \bibinfo{journal}{\emph{Management Science}} \bibinfo{volume}{68}, \bibinfo{number}{10} (\bibinfo{year}{2022}), \bibinfo{pages}{7112--7130}.
\newblock


\bibitem[Fibich et~al\mbox{.}(2003)]%
        {fibich2003explicit}
\bibfield{author}{\bibinfo{person}{Gadi Fibich}, \bibinfo{person}{Arieh Gavious}, {and} \bibinfo{person}{Oded Lowengart}.} \bibinfo{year}{2003}\natexlab{}.
\newblock \showarticletitle{Explicit solutions of optimization models and differential games with nonsmooth (asymmetric) reference-price effects}.
\newblock \bibinfo{journal}{\emph{Operations Research}} \bibinfo{volume}{51}, \bibinfo{number}{5} (\bibinfo{year}{2003}), \bibinfo{pages}{721--734}.
\newblock


\bibitem[Freund and Hssaine(2021)]%
        {freund2021fair}
\bibfield{author}{\bibinfo{person}{Daniel Freund} {and} \bibinfo{person}{Chamsi Hssaine}.} \bibinfo{year}{2021}\natexlab{}.
\newblock \showarticletitle{Fair incentives for repeated engagement}.
\newblock \bibinfo{journal}{\emph{arXiv preprint arXiv:2111.00002}} (\bibinfo{year}{2021}).
\newblock


\bibitem[Greenleaf(1995)]%
        {greenleaf1995impact}
\bibfield{author}{\bibinfo{person}{Eric~A Greenleaf}.} \bibinfo{year}{1995}\natexlab{}.
\newblock \showarticletitle{The impact of reference price effects on the profitability of price promotions}.
\newblock \bibinfo{journal}{\emph{Marketing science}} \bibinfo{volume}{14}, \bibinfo{number}{1} (\bibinfo{year}{1995}), \bibinfo{pages}{82--104}.
\newblock


\bibitem[Guo et~al\mbox{.}(2025)]%
        {guo2025multiproduct}
\bibfield{author}{\bibinfo{person}{Mengzi~Amy Guo}, \bibinfo{person}{Hansheng Jiang}, {and} \bibinfo{person}{Zuo-Jun~Max Shen}.} \bibinfo{year}{2025}\natexlab{}.
\newblock \showarticletitle{Multiproduct dynamic pricing with reference effects under logit demand}.
\newblock \bibinfo{journal}{\emph{Manufacturing \& Service Operations Management}} \bibinfo{volume}{27}, \bibinfo{number}{5} (\bibinfo{year}{2025}), \bibinfo{pages}{1645--1663}.
\newblock


\bibitem[Hu et~al\mbox{.}(2016)]%
        {hu2016dynamic}
\bibfield{author}{\bibinfo{person}{Zhenyu Hu}, \bibinfo{person}{Xin Chen}, {and} \bibinfo{person}{Peng Hu}.} \bibinfo{year}{2016}\natexlab{}.
\newblock \showarticletitle{Dynamic pricing with gain-seeking reference price effects}.
\newblock \bibinfo{journal}{\emph{Operations Research}} \bibinfo{volume}{64}, \bibinfo{number}{1} (\bibinfo{year}{2016}), \bibinfo{pages}{150--157}.
\newblock


\bibitem[Hu and Nasiry(2018)]%
        {hu2018markets}
\bibfield{author}{\bibinfo{person}{Zhenyu Hu} {and} \bibinfo{person}{Javad Nasiry}.} \bibinfo{year}{2018}\natexlab{}.
\newblock \showarticletitle{Are markets with loss-averse consumers more sensitive to losses?}
\newblock \bibinfo{journal}{\emph{Management Science}} \bibinfo{volume}{64}, \bibinfo{number}{3} (\bibinfo{year}{2018}), \bibinfo{pages}{1384--1395}.
\newblock


\bibitem[Jiang et~al\mbox{.}(2022)]%
        {jiang2022revenue}
\bibfield{author}{\bibinfo{person}{Bo Jiang}, \bibinfo{person}{Zizhuo Wang}, {and} \bibinfo{person}{Nanxi Zhang}.} \bibinfo{year}{2022}\natexlab{}.
\newblock \showarticletitle{Revenue Management Under a Price Alert Mechanism}.
\newblock \bibinfo{journal}{\emph{Available at SSRN 4154861}} (\bibinfo{year}{2022}).
\newblock


\bibitem[Jiang et~al\mbox{.}(2024)]%
        {jiang2024intertemporal}
\bibfield{author}{\bibinfo{person}{Hansheng Jiang}, \bibinfo{person}{Junyu Cao}, {and} \bibinfo{person}{Zuo-Jun~Max Shen}.} \bibinfo{year}{2024}\natexlab{}.
\newblock \showarticletitle{Intertemporal pricing via nonparametric estimation: Integrating reference effects and consumer heterogeneity}.
\newblock \bibinfo{journal}{\emph{Manufacturing \& Service Operations Management}} \bibinfo{volume}{26}, \bibinfo{number}{1} (\bibinfo{year}{2024}), \bibinfo{pages}{28--46}.
\newblock


\bibitem[Kahneman et~al\mbox{.}(1993)]%
        {kahneman1993more}
\bibfield{author}{\bibinfo{person}{Daniel Kahneman}, \bibinfo{person}{Barbara~L Fredrickson}, \bibinfo{person}{Charles~A Schreiber}, {and} \bibinfo{person}{Donald~A Redelmeier}.} \bibinfo{year}{1993}\natexlab{}.
\newblock \showarticletitle{When more pain is preferred to less: Adding a better end}.
\newblock \bibinfo{journal}{\emph{Psychological science}} \bibinfo{volume}{4}, \bibinfo{number}{6} (\bibinfo{year}{1993}), \bibinfo{pages}{401--405}.
\newblock


\bibitem[Liu(2023)]%
        {liu2023dynamic}
\bibfield{author}{\bibinfo{person}{Xiao Liu}.} \bibinfo{year}{2023}\natexlab{}.
\newblock \showarticletitle{Dynamic coupon targeting using batch deep reinforcement learning: An application to livestream shopping}.
\newblock \bibinfo{journal}{\emph{Marketing Science}} \bibinfo{volume}{42}, \bibinfo{number}{4} (\bibinfo{year}{2023}), \bibinfo{pages}{637--658}.
\newblock


\bibitem[Liu and Cooper(2015)]%
        {liu2015optimal}
\bibfield{author}{\bibinfo{person}{Yan Liu} {and} \bibinfo{person}{William~L Cooper}.} \bibinfo{year}{2015}\natexlab{}.
\newblock \showarticletitle{Optimal dynamic pricing with patient customers}.
\newblock \bibinfo{journal}{\emph{Operations research}} \bibinfo{volume}{63}, \bibinfo{number}{6} (\bibinfo{year}{2015}), \bibinfo{pages}{1307--1319}.
\newblock


\bibitem[Lobel(2020)]%
        {lobel2020dynamic}
\bibfield{author}{\bibinfo{person}{Ilan Lobel}.} \bibinfo{year}{2020}\natexlab{}.
\newblock \showarticletitle{Dynamic pricing with heterogeneous patience levels}.
\newblock \bibinfo{journal}{\emph{Operations Research}} \bibinfo{volume}{68}, \bibinfo{number}{4} (\bibinfo{year}{2020}), \bibinfo{pages}{1038--1046}.
\newblock


\bibitem[Mazumdar et~al\mbox{.}(2005)]%
        {mazumdar2005reference}
\bibfield{author}{\bibinfo{person}{Tridib Mazumdar}, \bibinfo{person}{Sevilimedu~P Raj}, {and} \bibinfo{person}{Indrajit Sinha}.} \bibinfo{year}{2005}\natexlab{}.
\newblock \showarticletitle{Reference price research: Review and propositions}.
\newblock \bibinfo{journal}{\emph{Journal of marketing}} \bibinfo{volume}{69}, \bibinfo{number}{4} (\bibinfo{year}{2005}), \bibinfo{pages}{84--102}.
\newblock


\bibitem[Nasiry and Popescu(2011)]%
        {nasiry2011dynamic}
\bibfield{author}{\bibinfo{person}{Javad Nasiry} {and} \bibinfo{person}{Ioana Popescu}.} \bibinfo{year}{2011}\natexlab{}.
\newblock \showarticletitle{Dynamic pricing with loss-averse consumers and peak-end anchoring}.
\newblock \bibinfo{journal}{\emph{Operations research}} \bibinfo{volume}{59}, \bibinfo{number}{6} (\bibinfo{year}{2011}), \bibinfo{pages}{1361--1368}.
\newblock


\bibitem[Popescu and Wu(2007)]%
        {popescu2007dynamic}
\bibfield{author}{\bibinfo{person}{Ioana Popescu} {and} \bibinfo{person}{Yaozhong Wu}.} \bibinfo{year}{2007}\natexlab{}.
\newblock \showarticletitle{Dynamic pricing strategies with reference effects}.
\newblock \bibinfo{journal}{\emph{Operations research}} \bibinfo{volume}{55}, \bibinfo{number}{3} (\bibinfo{year}{2007}), \bibinfo{pages}{413--429}.
\newblock


\bibitem[Puterman(2014)]%
        {puterman2014markov}
\bibfield{author}{\bibinfo{person}{Martin~L Puterman}.} \bibinfo{year}{2014}\natexlab{}.
\newblock \bibinfo{booktitle}{\emph{Markov decision processes: discrete stochastic dynamic programming}}.
\newblock \bibinfo{publisher}{John Wiley \& Sons}.
\newblock


\bibitem[Shmoys and Wang(2019)]%
        {shmoys2019how}
\bibfield{author}{\bibinfo{person}{David Shmoys} {and} \bibinfo{person}{Shujing Wang}.} \bibinfo{year}{2019}\natexlab{}.
\newblock \bibinfo{title}{How to solve a linear optimization problem on incentive allocation?}  (\bibinfo{year}{2019}).
\newblock
\urldef\tempurl%
\url{https://eng.lyft.com/how-to-solve-a-linear-optimization-problem-on-incentive-allocation-5a8fb5d04db1}
\showURL{%
\tempurl}
\newblock
\shownote{Lyft Engineering blog}.


\bibitem[Utley(2024)]%
        {alizila2024taobao_tmall_ai}
\bibfield{author}{\bibinfo{person}{Elizabeth Utley}.} \bibinfo{year}{2024}\natexlab{}.
\newblock \bibinfo{title}{Taobao and Tmall Upgrades Consumer Shopping Experience and Merchant Support Through AI}.  (\bibinfo{year}{2024}).
\newblock
\urldef\tempurl%
\url{https://www.alizila.com/taobao-and-tmall-upgrades-consumer-shopping-experience-and-merchant-support-through-ai/}
\showURL{%
\tempurl}
\newblock
\shownote{Alizila}.


\bibitem[Wang(2016)]%
        {wang2016intertemporal}
\bibfield{author}{\bibinfo{person}{Zizhuo Wang}.} \bibinfo{year}{2016}\natexlab{}.
\newblock \showarticletitle{Intertemporal price discrimination via reference price effects}.
\newblock \bibinfo{journal}{\emph{Operations research}} \bibinfo{volume}{64}, \bibinfo{number}{2} (\bibinfo{year}{2016}), \bibinfo{pages}{290--296}.
\newblock


\bibitem[Zhou(2023)]%
        {zhou2023temu_game}
\bibfield{author}{\bibinfo{person}{Viola Zhou}.} \bibinfo{year}{2023}\natexlab{}.
\newblock \bibinfo{title}{How Temu topped the U.S. app charts by turning shopping into a game}.  (\bibinfo{year}{2023}).
\newblock
\urldef\tempurl%
\url{https://restofworld.org/2023/temu-mobile-gaming/}
\showURL{%
\tempurl}
\newblock
\shownote{Rest of World}.


\end{thebibliography}

\appendix


\section{Processed Features for Prediction Model} \label{sec:practical_method}

\begin{table}[h]
\centering
\small
\setlength{\tabcolsep}{6pt}
\renewcommand{\arraystretch}{1.15}
\begin{tabularx}{\linewidth}{@{}>{\raggedright\arraybackslash}p{0.18\linewidth} X@{}}
\toprule
\textbf{Category} & \textbf{Features} \\
\midrule

Purchasing &
\begin{minipage}[t]{\linewidth}
Total purchase amount in the last 30/360 days\\
Number of transactions in the last 3/7/30/360 days\\
Average site sale discount of purchases in the last 30/360 days \\
Number of items in the cart\\
Total value of products in the cart\\
Average site sale discount in the cart
\end{minipage}
\\

\vspace{0.1\baselineskip}
Coupon &
\vspace{0.1\baselineskip}
\begin{minipage}[t]{\linewidth}
Number of free shipping coupons in last 7/28 days\\
Number of stackable coupons in last 7/28 days\\
Coupon value received 1/2 day ago\\
Largest coupon in the last 3/7/28 days\\
Average coupon in the last 3/7/28/30 days\\
Median coupon in the last 28 days\\
Variance of coupon values in the last 7/28 days\\
Average coupon clicked in the last 7/30 days\\
Average coupon opened in the last 7/30 days \\
Average coupon discount of purchases in the last 30/360 days \\
Average coupon discount of all historical purchases \\
Percentage of historical orders where a coupon was used
\end{minipage}
\\

\vspace{0.1\baselineskip}
Email &
\begin{minipage}[t]{\linewidth}
\vspace{0.1\baselineskip}
Email opened 1/2 day ago\\
Number of emails opened in last 3/7/28 days\\
Whether all emails were opened in the last 28 days\\
Email clicked 1/2 day ago\\
Number of emails clicked in last 3/7/28 days
\end{minipage}
\\

\vspace{0.1\baselineskip}
Website &
\begin{minipage}[t]{\linewidth}
\vspace{0.1\baselineskip}
Number of times viewed webpage in the last 1/3/7/30 days\\
Number of times viewed cart in the last 1/3/7/30 days\\
Number of days viewed cart in the last 1/3/7/30 days\\
Number of times viewed a product in the last 7 days\\
Number of times viewed checkout in the last 7 days
\end{minipage}
\\

\vspace{0.1\baselineskip}
Time &
\begin{minipage}[t]{\linewidth}
\vspace{0.1\baselineskip}
Day of week
\end{minipage}
\\

\bottomrule
\end{tabularx}
\caption{List of features used in the prediction model. Time windows written with slashes (e.g., 3/7/28 days) indicate multiple distinct features, one computed for each listed window length.}
\label{tab:all_features}
\end{table}

\section{Proof of Proposition~\ref{prop:simple}} \label{pf:simple}

Suppose $0\le\lambda\le\lambda'$ and $v,v'\in\cV$ with $v\le v'$.
If for customer $i$, the bigger discount $v'$ is preferred over $v$ under $\lambda'$, i.e.~$(1-\lambda' v')q(x_{it},v')\ge (1-\lambda' v)q(x_{it},v)$, then we have
$$
\frac{1-\lambda' v'}{1-\lambda' v}\ge\frac{q(x_{it},v)}{q(x_{it},v')}.
$$
We know however that $\frac{1-\lambda v'}{1-\lambda v}\ge\frac{1-\lambda' v'}{1-\lambda' v}$ by the rearrangement inequality, which implies that $(1-\lambda v')q(x_{it},v')\ge (1-\lambda v)q(x_{it},v)$, i.e.~$v'$ is also preferred over $v$ under $\lambda$.  Therefore, increasing $\lambda_t$ from $\lambda$ to $\lambda'$ cannot cause $v_{it}$ to increase for any customer $i$.

Under the assumption that $q(x_{it},v)$ is increasing in $v$, it is clear that~\eqref{eqn:discRedeemed} is increasing in $v_{it}$ for all $i$, and hence decreasing in $\lambda_t$.

\section{Proof of Theorem~\ref{thm:tightness}} \label{pf:tightness}

We can without loss assume $\cP=\{\rho_0,\ldots,\rho_{d-1}\}$, by creating arbitrarily negative values for $g(r,p)$ when $p\notin\{\rho_0,\ldots,\rho_{d-1}\}$.  For convenience, we can denote the price set to be $\cP=\{1,\ldots,d\}$.

Take any real values $h(1)<\cdots<h(d)$.  Define
\begin{align} \label{eqn1}
g(\rho_t) &:=\frac{h(\rho_{t-1\mod d})-h(\rho_t)}{k(\rho_{t-1\mod d},\rho_t)}+C &\forall t=0,\ldots,d-1,
\end{align}
where $k(\cdot,\cdot)$ is defined as in the statement of \Cref{thm:tightness}, and $C$ is a large positive constant that ensures $g(\rho_t)>0$ for all $t=0,\ldots,d-1$ (it suffices if $\frac{h(1)-h(d)}{\ell}+C>0$).
Finally, define the reference-monotone gain function $g$ to be
\begin{align} \label{eqn2}
g(r,\rho_t)
&=g(\rho_t)\bI(r\ge\rho_{t-1\mod d})
&\forall t=0,\ldots,d-1; r\in\cP.
\end{align}

We show for this gain function $g$ that the unique optimal price cycle is $\rho_0^{k(\rho_{d-1},\rho_0)}\rho_1^{k(\rho_{0},\rho_1)}\cdots\rho_{d-1}^{k(\rho_{d-2},\rho_{d-1})}$, as required for the statement of \Cref{thm:tightness}.
It suffices to show that if substitute these values of $h(1),\ldots,h(d)$ into the optimality condition~\eqref{eqn:bellman}, along with $\OPT=C$ (the long-run average gain of the optimal price cycle), then the "max" is achieved if and only if $r=\rho_{t-1\mod d},p=\rho_t$ for some $t=0,\ldots,d-1$.  In other words, we need to prove
\begin{align*}
h(\rho_{t-1\mod d}) &\ge (g(\rho_{t-1\mod d},\rho_{t'})-C)k(\rho_{t-1\mod d},\rho_{t'})+h(\rho_{t'}) &\forall t,t'\in\{0,\ldots,d-1\}
\end{align*}
or equivalently
\begin{align} \label{eqn:toShow}
g(\rho_{t-1\mod d},\rho_{t'}) &\le \frac{h(\rho_{t-1\mod d})-h(\rho_{t'})}{k(\rho_{t-1\mod d},\rho_{t'})}+C  &\forall t,t'\in\{0,\ldots,d-1\}
\end{align}
with equality if and only if $t'=t$.

If $t'=t$, then equality holds by the definitions in~\eqref{eqn1}--\eqref{eqn2}.
Now suppose $t'\neq t$, and first consider the case where $\rho_{t-1\mod d}<\rho_{t'-1\mod d}$.
We have
\begin{align*}
g(\rho_{t-1\mod d},\rho_{t'})
=0
<\frac{h(1)-h(d)}{\ell}+C
\le\frac{h(\rho_{t-1\mod d})-h(\rho_{t'})}{k(\rho_{t-1\mod d},\rho_{t'})}+C
\end{align*}
as desired, where the final (weak) inequality holds because the smallest possible value of $\frac{h(\rho_{t-1\mod d})-h(\rho_{t'})}{k(\rho_{t-1\mod d},\rho_{t'})}$ is $\frac{h(1)-h(d)}{\ell}$.
In the other case where $\rho_{t-1\mod d}>\rho_{t'-1\mod d}$, we have
\begin{align*}
g(\rho_{t-1\mod d},\rho_{t'})=g(\rho_{t'})
=\frac{h(\rho_{t'-1\mod d})-h(\rho_{t'})}{k(\rho_{t'-1\mod d},\rho_{t'})}+C
<\frac{h(\rho_{t-1\mod d})-h(\rho_{t'})}{k(\rho_{t-1\mod d},\rho_{t'})}+C
\end{align*}
as desired, where the inequality holds because $\frac{h(p)-h(\rho_{t'})}{k(p,\rho_{t'})}$ is a strictly increasing function over the prices $p\in\cP$ and $\rho_{t-1\mod d}>\rho_{t'-1\mod d}$.
This completes the proof of \Cref{thm:tightness}.

\end{document}